\documentclass[12pt, draftclsnofoot, onecolumn]{IEEEtran}
\IEEEoverridecommandlockouts
\usepackage{graphicx}
\usepackage{float}
\usepackage[cmex10]{amsmath}
\usepackage[compress]{cite}
\usepackage{balance}
\usepackage[UKenglish]{babel}
\usepackage{amssymb}
\usepackage{graphicx, picture, calc}
\usepackage{epstopdf}
\usepackage{subcaption}
\usepackage{mathtools}
\usepackage{booktabs}
\usepackage[compress]{cite}
\usepackage{enumerate}
\usepackage[norule]{footmisc}
\usepackage{algorithm}
\usepackage{algpseudocode}
\usepackage{amsmath}
\usepackage{longtable}
\usepackage{dsfont}
\usepackage{pifont}
\usepackage{color}
\usepackage{xcolor}
\usepackage{amsthm}
\usepackage{multirow}
\usepackage[linesnumbered, algo2e, ruled,norelsize]{algorithm2e}
\usepackage[noabbrev]{cleveref}

\makeatother

\newtheorem{Proposition}{\bf{Proposition}}

\def\BibTeX{{\rm B\kern-.05em{\sc i\kern-.025em b}\kern-.08em
    T\kern-.1667em\lower.7ex\hbox{E}\kern-.125emX}}

\begin{document}
\title{Design of Energy-Efficient Artificial Noise for Physical Layer Security in  Visible Light Communications}
\author{
        Thanh~V.~Pham,~\IEEEmembership{Member,~IEEE,}
         Anh~T.~Pham,~\IEEEmembership{Senior~Member,~IEEE,}% <-this % stops a space
         ~and 
        Susumu~Ishihara,~\IEEEmembership{Member,~IEEE}
\thanks{Thanh V. Pham and Susumu Ishihara are with the Department of Mathematical and Systems Engineering, Shizuoka University, Shizuoka, Japan (e-mail: pham.van.thanh@shizuoka.ac.jp, susumu.ishihara@shizuoka.ac.jp).}

\thanks{Anh T. Pham is with the Department of Computer Science and Engineering, The University of Aizu,
Fukushima, Japan (e-mail: pham@u-aizu.ac.jp).}% <-this % stops a space
\thanks{Part of this paper has been presented at the 2021 IEEE International Conference on Communications (ICC 2021), Workshop on Green Solutions for Smart Environment, Montreal, Canada. }
}

\maketitle

\begin{abstract}
This paper studies the design of energy-efficient artificial noise (AN) schemes in the context of physical layer security in visible light communications (VLC). Two different transmission schemes termed \textit{selective AN-aided single-input single-output (SISO)} and \textit{AN-aided multiple-input single-output (MISO)} are examined and compared in terms of secrecy energy efficiency (SEE). In the former, the closest LED luminaire to the legitimate user (Bob) is the information-bearing signal's transmitter. At the same time, the rest of the luminaries act as jammers transmitting AN to degrade the channels of eavesdroppers (Eves). In the latter, the information-bearing signal and AN are combined and transmitted by all luminaries. When Eves' CSI is unknown, an indirect design to improve the SEE is formulated by maximizing Bob's channel's energy efficiency. A low-complexity design based on the zero-forcing criterion is also proposed.  In the case of known Eves' CSI, we study the design that maximizes the minimum SEE among those corresponding to all eavesdroppers. At their respective optimal SEEs, simulation results reveal that when Eves' CSI is unknown, the selective AN-aided SISO transmission can archive twice better SEE as the AN-aided MISO does. In contrast, when Eves' CSI is known, the AN-aided MISO outperforms by 30\%. 
\end{abstract}

\begin{IEEEkeywords}
VLC, energy efficiency, physical layer security, jamming, precoding.  
\end{IEEEkeywords}

\section{Introduction}
\subsection{Background}
Indoor wireless activities generate roughly 80\% of the total wireless Internet traffic \cite{chun2019a}.
In the case of indoor scenarios, a standard means of Internet access is Wireless Fidelity (Wi-Fi), which is currently imposed by increasingly stringent requirements for high data rates due to  
the tremendous growth in the number of mobile devices and data-intensive applications. A great deal of research effort has been devoted to satisfying this increasing demand by exploring new communications technologies, among which visible light communications (VLC) is a promising candidate. Compared to existing wireless systems, VLC offers the unique advantage of simultaneously providing illumination and communication, both essential for indoor activities. Furthermore, the promised high data rate comes at no license-fee spectrum and little infrastructural installation.

% During the long history of development, a great deal of research has been paid to improve the reliability and spectral efficiency of communications systems through sophisticated coding and modulation schemes. However, over the past decade, there has been an increasing interest in energy-efficient solution as part of the green communication. This is due to the current global effort in reducing greenhouse gas emission. It is estimated that the telecommunication industry consumes around 3\% of global energy and its carbon footprint represents approximately 2\% of global emission. Moreover, it is reported that approximated 80\% of wireless data traffic happen indoors where illumination is often necessary.    

Over the years, numerous theoretical and experimental studies on VLC have focused on improving its practicality and data-rate performance (see \cite{Matheus2019,Ahmed2020,Amjad2021} and references therein) for both indoor and outdoor scenarios. To realize practical VLC systems, one must also incorporate security and privacy in the design. Traditional approaches for acquiring secure transmissions over public communication media are based on well-established key-based cryptographic algorithms. The security of conventional cryptography relies on the extreme computational complexity of solving certain mathematical problems (e.g., the factoring problem). Given the current computing power of the classical computer, it is likely that breaking the secret keys is not feasible within a meaningful amount of time. However, quantum computers, which are expected to be fully functional in the foreseeable future, can threaten the security of current encryption schemes as they are exponentially faster than their classical counterparts \cite{Davide2022}. The vulnerability of traditional cryptography to quantum computers has been the primary motivation for researching alternative security measures not based on computational difficulty. In this context, physical layer security (PLS) is a promising approach to complement and/or replace cryptographic techniques as the information confidentiality can be kept perfectly secure regardless of the computational power of the eavesdropper\cite{Wyner1975}. 
\subsection{Related Works}
While the literature on PLS in VLC systems is considerably extensive (see \cite{BLINOWSKI2019246,Arfaoui2020Survey} and references therein), the existing studies mainly focused on two directions. The first focuses on analyzing different lower and upper bounds on the secrecy capacity of the single-input single-output (SISO) wiretap VLC channels \cite{Wang2018,Wang2019,Wang2021}. 
Compared with the SISO, the multiple-input single-output (MISO) configuration is more relevant in practice, as deploying multiple LED luminaires is a norm in illumination. 
In the second direction, several studies, therefore, exploited the spatial degree of freedom at the transmitting side to investigate the application of precoding and artificial noise (AN) to improve secrecy performance. Specifically, in the case of precoding, Mostafa $et~al.$ studied the designs of zero-forcing (ZF) and general precoders to maximize the achievable secrecy rate for systems with a single legitimate user and a single eavesdropper considering both perfect and uncertain estimation of the channel state information (CSI) \cite{mostafa2015physical,mostafa2016optimal}. Subsequent works then examined precoding designs for more general system configurations with different objectives, such as multi-eavesdropper with minimizing the total transmitted power \cite{ma2016optimal}, multi-legitimate users with maximizing the achievable secrecy sum-rate \cite{pham2017secrecy}, and multi-user systems, in which each user treats others as eavesdroppers (hence, transmitted messages for users must be kept mutually confidential)\cite{arfaoui2018secrecy}. ZF precoding strategies to deal with both active and passive eavesdroppers were also proposed in \cite{Cho2021}.

% and improving secrecy performance. In fact, to the best of our knowledge, our previous research in \cite{Son2021,Duong2021} are the first works investigating PLS in VLC from the EE perspective. Specifically, we designed energy-efficient precoding schemes for multi-user VLC systems where users' messages are kept mutually confidential. 

%Note that the use of precoding is efficient in the case that the channel state information (CSI) of legitimate users and eavesdroppers is known by the transmitter \cite{pham2017secrecy,Arfaoui2018,Cho2021}. This is a relatively strong assumption since, in practice, eavesdroppers are usually passive, thus, do not provide their CSI to the transmitting side. 

In addition to precoding, AN is another approach to enhancing PLS performance \cite{Goel2008,Zou2016,Liu2017}.
%In such a scenario, the use of artificial noise (AN) can be more appropriate to enhance the PLS performance \cite{Goel2008,Zou2016,Liu2017}. 
In essence, AN is a jamming signal which is purposely generated and transmitted by the transmitter so that it causes minimal or no interference to the legitimate user while potentially degrading the quality of the eavesdroppers' channel. Note that AN should be deployed with precoding to fully exploit the available spatial degree of freedom. Hence, precoding can be considered a special case of AN-aided transmission where no power is allocated to AN generation. Initial research on the topic considered systems where a fixed LED luminaire transmits the information-bearing signal while the rest broadcast jamming signals \cite{Mostafa2014,Zaid2015}. We term this scheme as \emph{AN-aided SISO transmission}. Subsequent works then focused on the jamming scheme where the information-bearing signal and AN are combined and transmitted by all luminaires. This scheme is referred to as \emph{AN-aided MISO transmission}. Joint precoding and AN designs were investigated under different performance criteria. For example, the authors in \cite{Shen2016} studied the design that maximizes the signal-to-interference-plus-noise ratio (SINR) of the legitimate user's channel while constraining that of the eavesdroppers' channel to a predefined threshold. The study in \cite{Cho2019} examined the same objective yet under the constraint on the average SINR of the eavesdroppers' channel (assuming that eavesdroppers' CSI is unknown at the transmitters). Considering a multi-Bob VLC system, AN designs to maximize the minimum SINR among Bobs' channels were investigated in \cite{Pham2019Acess}. Given specific constraints on the SINRs of Bob's and Eves' channels, our study in \cite{pham2020energy} AN designs to minimize the total transmitted power considering both perfect and imperfect CSI estimations.

Aside from improving the secrecy rate, there is also a need to optimize energy efficiency (EE), which measures how much energy is consumed per a transmitted information bit \cite{Bjornson2018}. This has been of particular interest in recent years due to the increased effort in reducing energy consumption to combat climate change \cite{Scheck2010}. Although several studies have been conducted on EE in VLC and hybrid VLC/RF systems \cite{Kashef2016,Khreishah2018,Zhang2018,Ma2018,Aboagye2020,An2020}, little attention has been paid to that in the context of PLS. In fact, to the best of the authors' knowledge, our previous studies in \cite{Son2021,Duong2021} are the first works concerning the PLS in VLC from the EE perspective. Defining secrecy energy efficiency (SEE) as the consumed energy per a transmitted confidential bit, we designed precoding schemes to maximize the sum SEE for multi-user VLC systems where messages among users are kept mutually confidential.     
\subsection{Contributions}
Against the above background, this paper aims to study energy-efficient AN designs for PLS in VLC systems. We consider the systems with a legitimate user (i.e., Bob) and multiple eavesdroppers (i.e., Eves) whose CSI can be known or unknown on the transmitting side. The AN-aided MISO and a refined AN-aided SISO transmission scheme, which we term as \emph{selective AN-aided SISO}, are investigated and compared in terms of the achievable SEE. 
Contrary to the original AN-aided SISO, where the information-bearing signal is transmitted by a fixed LED luminaire, in the proposed scheme, that is done by the luminaire, which has the highest channel gain to Bob. The rest of the luminaires act as jammers transmitting AN. By doing so, Bob can enjoy the strongest information-bearing signal while suffering less from the AN. Hence, given the same amount of consumed energy, the proposed scheme should achieve a better achievable secrecy rate, resulting in a higher SEE compared with the original counterpart. In the conference version of this paper \cite{Pham2021}, the design of the selective AN-aided SISO scheme considering the unavailability of Eves' CSI was studied. Extending from this, the contributions of this paper are summarized as follows. 

\begin{itemize}
    \item In addition to the presented design in \cite{Pham2021}, which was solved using the Dinkelbach algorithm and the convex-concave procedure (CCP), a simple and low-complexity ZF AN-aided SISO design is investigated. Simulation results show that this design approach achieves comperable SEE that of the design in \cite{Pham2021} with significantly reduced computational time. 
    \item We study an AN design when Eves' CSI is available at the transmitting side. The presence of multiple Eves results in multiple sub-wiretap channels, each characterized by a different achievable secrecy rate. To ensure that the transmitted information to Bob is secure against all Eves, the system's achievable secrecy rate is defined as the minimum among all sub-wiretap channels. Then, an AN design is cast as a max-min SEE problem, which requires a different approach to solve. Specifically, we transform the original problem into a feasibility problem, which can then be solved using the bisection method. 
    \item Comprehensive simulations are conducted to compare the complexity and performance of the original AN-aided SISO, the proposed selective AN-aided SISO, and the AN-aided MISO transmissions under various parameter settings. Our main finding is that the selective AN-aided SISO scheme achieves a better SEE than the AN-aided MISO does in the case of unknown Eves' CSI, while it is the opposite in the case of known Eves' CSI as the AN-aided MISO scheme is superior. 
\end{itemize}

%different energy-efficient AN designs in the context of PLS are studied. Specifically, two AN transmission schemes termed \textit{selective} and \textit{full transmission} are investigated. The latter was, in fact, considered in most of the previous works.   
%As a matter of fact that the strongest signal to Bob, which contributes the majority to improving its EE, comes from the closest LED luminaire. Therefore, our proposal is that the closest LED luminaire  to Bob transmits an information-bearing signal where the rest of the luminaries emit AN. The use of multiple luminaries as jammers is to increase the probability that Eves are confused by AN. The design problem is non-convex fractional programming, thus rendering finding the optimal solution difficult. To tackle the problem, we employ a combination of the Dinkelbach algorithm and convex-concave procedure (CCP), which guarantees a local solution.      
\subsection{Organization}
The rest of the paper is organized as follows. The system models, including the channel model, signal model, and energy consumption, are described in Section II. Section III presents different AN-aided designs from the perspective of SEE maximization. Numerical results with related discussions are given in Section IV. Finally, Section V concludes the paper. 

\emph{Notation}: The following notations are used in the paper. $\mathbb{R}$ and $\mathbb{R}_{\geq 0}$ denote the set of real and nonnegative real numbers, respectively. Uppercase and lowercase bold letters (e.g., $\mathbf{A}$ and $\mathbf{a}$) represent matrices and column vectors, respectively. $\mathbf{A}^T$ is the transpose of $\mathbf{A}$. In addition, $|\cdot|$, $\lVert\cdot\rVert_2$,  $\lVert\cdot\rVert_{\infty}$, and $\text{tr}(\cdot)$ are the absolute value, Euclidean norm, maximum norm, and trace operators, respectively. Finally, $\mathbb{E}[x]$ is the expected value of $x$, $[\mathbf{a}]_n$ denotes the $n$-th element of $\mathbf{a}$ and $\mathbf{1}_{N}$ is the all-one column vector of size $N$.   
%-------------------------%-------------------------%-----------------
\section{System Models}
We consider a typical room as depicted in Figs.~\ref{sysmodel1} and \ref{sysmodel2}, where $N_T$ LED luminaries are deployed for both illumination and communications. These LEDs are connected via wired connections to a central processing unit (CPU), which is responsible for signal processing and coordination among LED transmitters. There are one legitimate user (Bob) and multiple eavesdroppers (Eves), who, assuming that, are non-colluding\footnote{In practice, Eves can cooperate to maximize their eavesdropped information. This, however, may complicate the analysis of the AN design. Hence, we leave this scenario for future research.}. It is assumed Bob feeds its estimated CSI back to the CPU via a wireless uplink (e.g., infrared or Wi-Fi). By comparing the entries of the received CSI vector, the CPU can determine the closest luminaire to Bob. 
In the selective AN-aided SISO transmission scheme illustrated in Fig.~\ref{sysmodel1}, the closest luminaire to Bob acts as Alice, who transmits the information-bearing signal. The rest of the luminaries act as jammers transmitting AN. In the AN-aided MISO  scheme described in Fig.~\ref{sysmodel2}, all luminaries act as Alice and jammer to transmit the combined information-bearing signal and AN. We also consider that Eves's CSI can be either known or unknown on the transmitting side. Accordingly, different design approaches are presented considering the availability of Eves' CSI.   
\begin{figure}
\centering
\begin{subfigure}[t]{0.49\textwidth}
\centering
\includegraphics[scale=0.36]{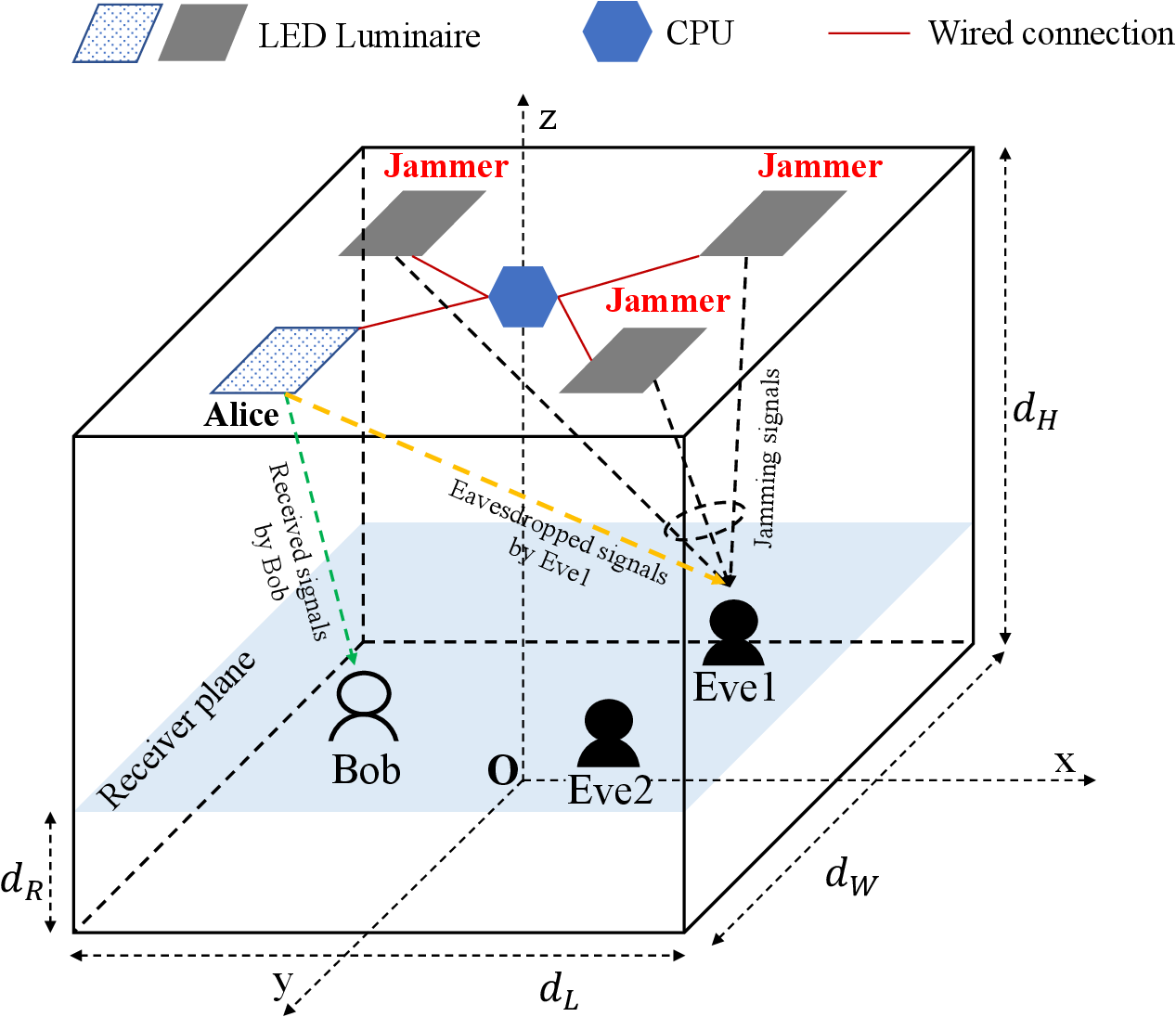}
\caption{Selective AN-aided SISO transmission.}
\label{sysmodel1}
\end{subfigure}
\begin{subfigure}[t]{0.49\textwidth}
\centering
\includegraphics[scale = 0.36]{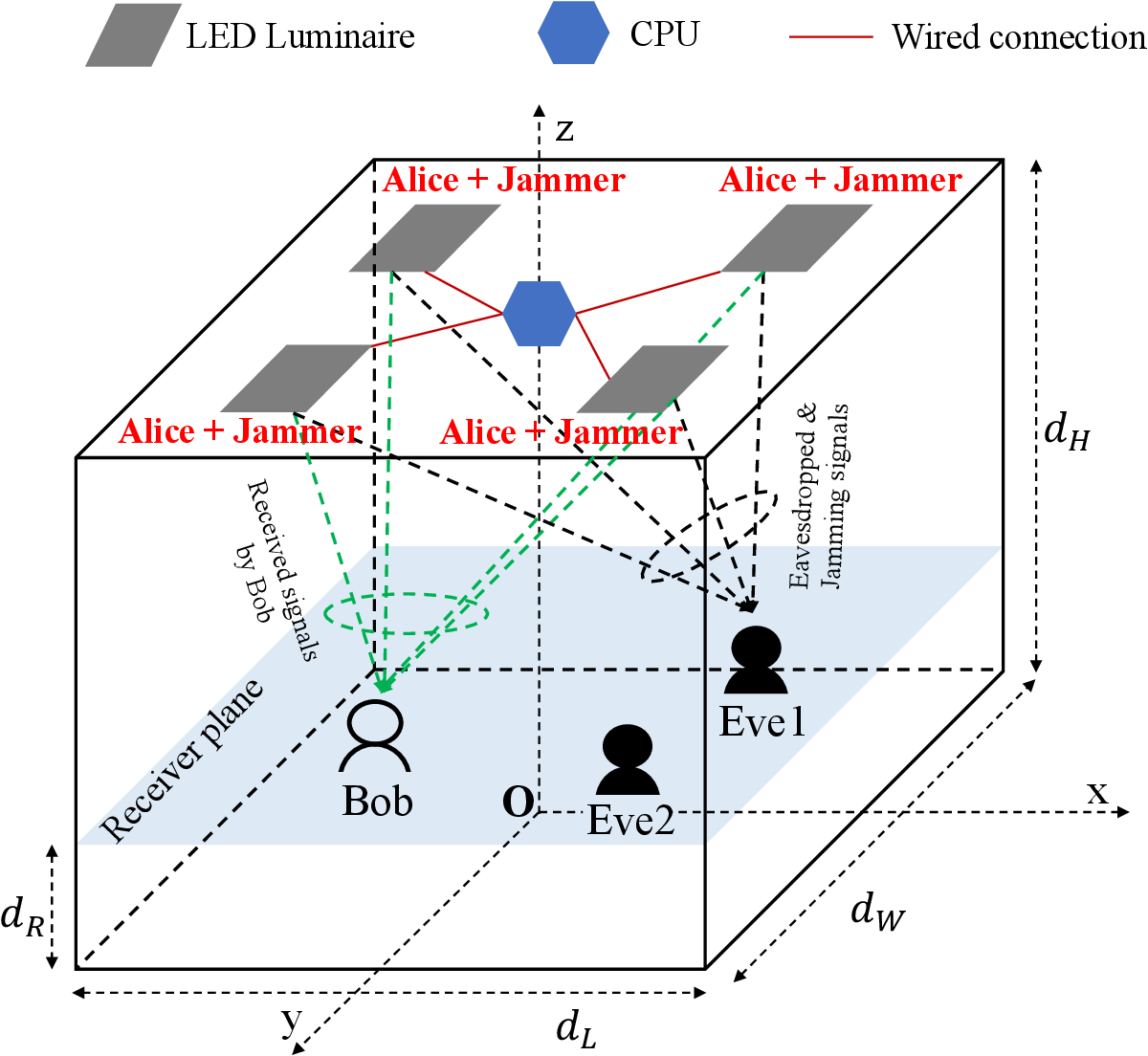}
\caption{AN-aided MISO transmission.}
\label{sysmodel2}
\end{subfigure}
\caption{AN-aided transmission schemes.}
\end{figure}
%all luminaries act as jammers, transmitting AN to prevent Eves from eavesdropping on the transmission to Bob. The closest luminaire  to Bob (called Alice), in addition to sending AN, transmits information-bearing signal to Bob. Realistically, Eves are passive users whose presences are not known to the transmitter (i.e., Alice and jammers). It is hence assumed that Eves' CSIs are unknown to Alice and jammers. 
\subsection{Channel Model}
In indoor environments, the transmitted signal propagates to a user via a direct light-of-sight (LoS) and multiple non-light-of-sight paths (NLoS) (due to reflections of walls and ceiling). Nevertheless, it was experimentally verified that the total received optical power at the receiver is predominantly contributed by the LoS signal. On the other hand, the NLoS channel gain, which accounts for less than 5\% of the total received optical power, 
Therefore, only the LoS propagation path is considered in this work to simplify the analysis.

The LoS channel coefficient between an LED luminaire  and a user denoted as $h$ is given by \cite{Komine2004}
\begin{align}
h  = & \frac{A_r(m+1)}{2\pi l^2}T_s(\psi)g(\psi)\cos^{m}(\phi)\cos(\psi)  \times \mathds{1}_{[0, \Psi]}(\psi), 
\label{eqn:chann_coeff}
\end{align}
where $\mathds{1}_{[x, y]}(\cdot)$ denotes the indicator function, $A_r$ is the active area of the photodiode (PD), $l$ is the link length,  $\phi$ is the angle of irradiance, $\Psi$ is the optical field of view (FOV) of the PD, and $\psi$ is the angle of incidence. Also, $m=-\frac{\ln(2)}{\ln(\Theta_{0.5})}$ is the order of the Lambertian emission 
% and is determined by 
% \begin{equation}
%     \label{eqn:order_Lambertian}
%     m=-\frac{\ln(2)}{\ln(\Theta_{0.5})},
% \end{equation}
where $\Theta_{0.5}$ is the LED's semi-angle for half illuminance. $T_s(\psi)$ is the gain of the optical filter, and $g(\psi) = \frac{\kappa^2}{\sin^2(\Psi)}\mathds{1}_{[0, \Psi]}(\psi)$ is the gain of the optical concentrator, 
% which is given by
% \begin{align}
% \label{eqn:gain_optical_concentrator}
%     g(\psi) = \frac{\kappa^2}{\sin^2(\Psi)}\mathds{1}_{[0, \Psi]}(\psi),
% %    \begin{cases}
% %        \frac{\kappa^2}{\sin^2(\Psi)} & 0 \leq \psi_{n,k} \leq \Psi, \\
% %        0 & \psi_{n,k} > \Psi, \\
% %    \end{cases}
% \end{align}
with $\kappa$ being the refractive index of the optical concentrator.
%$L(\phi)$ is the emission intensity of Lambertian light source, which is calculated as 
%\begin{equation}
%    \label{eqn:Lambertian}
%    L(\phi)=\frac{l+1}{2\pi}\cos^l(\phi),
%\end{equation}
%where $l$ is the order of Lambertian emission determined by
\subsection{Signal Model}
In this section, only the signal model of the selective AN-aided transmission scheme is presented for the sake of conciseness. Without otherwise noted, analysis in the case of the AN-aided MISO transmission scheme follows the same manner. In the following, signal models for the case of unknown and known Eves's CSI are respectively described. 
\subsubsection{Unknown Eves' CSI}
Let $d \in \mathbb{R}$ be a $M$-ary pulse amplitude modulation (M-PAM)-modulated information signal intended for Bob. Also, let $\mathbf{z} \in \mathbb{R}$ be a random signal representing the AN. Without loss of generality, it is assumed that $d$ and $z$ are both zero-mean and normalized and uniformly distributed over $[-1,~1]$, which is notated as $d, z \sim \mathcal{U}[-1, 1]$. In our proposed scheme, $d$ is sent by Alice while $\mathbf{z}$ is transmitted by the jammers. For Alice, the information-bearing signal $d$ is first scaled by a weighting factor (i.e., precoder) $v$ and then added by a DC-bias $I_{\text{DC}}$, which specifies the desired illumination level. The resulting LED's modulating signal is then represented by
\begin{align}
    x_d = vd + I_{\text{DC}}.
    \label{modulatingCurrent}
\end{align}

It is known that there exists a linear range (denoted as, for example, $[I_{\text{min}},~I_{\text{max}}]$) for each particular LED in which the output optical signal is linearly proportional to the amplitude of the modulating signal $x_d$. Therefore, to ensure efficient use of the LEDs,  $x_d$ should be constrained within the linear range as
\begin{align}
    I_{\text{min}} \leq x_d \leq I_{\text{max}}.
    \label{constratinOnModulatingCurrent}
\end{align}

As $d \sim \mathcal{U}[-1,~1]$, the following constraint on $v$ is needed to satisfy \eqref{constratinOnModulatingCurrent}
\begin{align}
    |v| \leq \Delta_{\text{DC}}, 
    \label{precoder-constraint-1}
\end{align}
where $\Delta_{\text{DC}} = \text{min}\left(I_{\text{DC}} - I_{\text{min}}, I_{\text{max}} - I_{\text{DC}}\right)$. When information regarding Eves (e.g., the number of Eves and their positions) is unknown at the transmitters, the use of a single precoder $\mathbf{w} \in \mathbb{R}^{N_T -1}$ for the AN is appropriate to simplify the design. Similar to \eqref{modulatingCurrent}, it follows that
\begin{align}
    \lVert\mathbf{w}\rVert_{\infty}  \leq \Delta_{\text{DC}}.
    \label{precoder-constraint-2}
\end{align}

Let $h_{\text{R}} \in \mathbb{R}$ and $\overline{\mathbf{h}}_{\text{R}} \in \mathbb{R}^{N_T-1}$ be the channel gain and channel gain vector from Alice and jammers to Bob, respectively\footnote{The subscript `R' is used to denote Bob and Eves as receivers in general. When necessary, the subscript `B' is used to refer to Bob while `E' is used to refer to Eve.}. If we denote $p_{d} = \eta\left(vd + I_{\text{DC}}\right)$ and $\mathbf{p}_{z} = \eta\left(\mathbf{w}z + I_{\text{DC}}\right)$ with $\eta$ being the LED conversion factor as the emitted optical power at Alice and jammers, respectively,
the received current signals at Bob and Eves can be written as 
\begin{align}
    y_{\text{R}}  = \gamma \mathbf{h}^T_{\text{R}} \mathbf{p}_{\text{t}} + n_{\text{R}} 
                  = \gamma \eta\left(h_{\text{R}}vd + \overline{\mathbf{h}}^T_{\text{R}}\mathbf{w}z + \mathbf{h}^T_{\text{R}}\textbf{I}_{\text{DC}}\right) + n_{\text{R}}, 
                  \label{receivedSignal}
\end{align}
where $\gamma$ is the photo-diode responsivity, $\mathbf{h}_{\text{R}} = \begin{bmatrix}h_{\text{R}} & \overline{\mathbf{h}}^T_{\text{R}}\end{bmatrix}^T$, $\mathbf{p}_{\text{t}} = \begin{bmatrix}p_{\text{d}}& \mathbf{p}_{z}^T\end{bmatrix}^T$, $\mathbf{I}_{\text{DC}} = I_{\text{DC}}\mathbf{1}_{N_T}$ is the DC-bias vector, and $n_{\text{R}}$ is the receiver noise. It is reasonable to model $n_{\text{R}}$ as a zero-mean additive white Gaussian noise with the variance being given by
\begin{align}
    \sigma^2_{\text{R}} = 2\gamma e\overline{p}^r_{\text{R}}B_{\text{mod}} + 4\pi e A_{\text{d}}\gamma\chi_{\text{amb}}(1-\cos(\Psi))B + i^2_{\text{amp}}B_{\text{mod}}, 
\end{align}
where $e$ is the elementary charge, $\overline{p}^r_{\text{R}} = \mathbb{E}\left[p^r_{\text{R}}\right] = \eta\mathbf{h}^T_{\text{R}}\textbf{I}_{\text{DC}}$ is the average received optical power, $B_{\text{mod}}$ is the modulation bandwidth of the LED, $\chi_{\text{amb}}$ is the ambient light photo-current, and $i_{\text{amp}}$ is the pre-amplifier noise current density. 

At the receiver, the DC current $\mathbf{h}^T_{\text{R}}\mathbf{I}_{\text{DC}}$ is filtered out, leaving the AC term for signal demodulation 
\begin{align}
    \overline{y}_{\text{R}}= \gamma\eta\left(h_{\text{R}}vd + \overline{\mathbf{h}}^T_{\text{R}}\mathbf{w}z\right) + n_{\text{R}}.
    \label{channelModel}
\end{align}

The signal-to-interference-plus-noise ratios (SINRs) of Bob's and Eves' received signals are then calculated as
\begin{align}
    \text{SINR}_{\text{R}} = \frac{\left(\gamma\eta \mathbb{E}[d^2]h_{\text{R}}v\right)^2}{\left(\gamma\eta \mathbb{E}[z^2]\overline{\mathbf{h}}^T_{\text{R}}\mathbf{w}\right)^2 + {\sigma}^2_{\text{R}}} = \frac{\left(h_{\text{R}}v\right)^2}{\left(\overline{\mathbf{h}}^T_{\text{R}}\mathbf{w}\right)^2 + \overline{\sigma}^2_{\text{R}}},
\end{align}
where $\overline{\sigma}^2_{\text{R}} = \frac{\sigma^2_{\text{R}}}{\frac{1}{3}\left(\gamma\eta\right)^2}$ (since $d, z \sim \mathcal{U}[-1, 1]$, then $\mathbb{E}[d^2] = \mathbb{E}[z^2] = \frac{1}{3}$).
%-------------------%-------------------
\subsubsection{Known Eves' CSI}
In certain scenarios, the information intended for Bob must be kept confidential; thus, any other active users should be treated as Eves. 
As a result, Eves' CSI can be assumed to be known on the transmitting side. To better exploit the available spatial information, a separate AN signal should be generated for each Eve. Thus, 
denote $z_k \sim \mathcal{U}[-1, 1]$ as the AN symbol intended for the $k$-th Eve and $\mathbf{w}_k \in \mathbb{R}^{N_T-1}$ as the corresponding precoder. The LED's modulating current vector at the jammers is then written as
\begin{align}
    \mathbf{x}_z = \sum_{k = 1}^K\mathbf{w}_kz_k + I_{\text{DC}},
\end{align}
where $K$ is the number of Eves. Bounding $\mathbf{x}_z$ between $\left[I_{\text{min}},~I_{\text{max}}\right]$ results in the following constraint on $\mathbf{w}_k$'s
\begin{align}
    \sum_{k = 1}^K\left|[\mathbf{w}_k]_n\right|\leq \Delta_{\text{DC}},~~~\forall n = 1, 2, \hdots, N_T-1.
    \label{constraint_AN_precoder}
\end{align}

The received AC current signals at Bob and Eves are then given by
\begin{align}
    y_{\text{R}} = \gamma\mathbf{h}^T_{\text{R}}\mathbf{p}_{t} + n_{\text{R}} = \gamma\eta\left(h_{\text{R}}vd + \overline{\mathbf{h}}^T_{\text{R}}\sum_{k = 1}^K\mathbf{w}_kz_k \right) + n_{\text{R}},
    \label{receivedSignalKnownEve}
\end{align}
where in this case $\mathbf{p}_{\text{t}} = \begin{bmatrix}p_d & \mathbf{p}^T_z\end{bmatrix}^T$ with $\mathbf{p}_z = \eta\left(\sum_{k = 1}^K\mathbf{w}_kz_k + I_{\text{DC}}\right)$.
%-------------------%-------------------
\subsubsection{Signal Model for the AN-aided MISO}
In the case of AN-aided MISO, since each luminaire acts as Alice and jammer simultaneously, notations of the channel and the precoder differ slightly from those of the selective AN-aided SISO. These differences are summarized in Table \ref{table1}. 
\begin{table}[ht]
\centering
\caption{Notational differences between Selective AN-aided SISO and AN-aided MISO.}
\label{table1}
\resizebox{0.8\textwidth}{!}{\begin{tabular}{|l|l|l|}
 \hline
  & Selective AN-aided SISO & AN-aided MISO \\ 
\hline Alice-Receiver's channel & $h_{\text{R}} \in \mathbb{R}_{\geq 0}$ &  \multirow{2}{*}{$\mathbf{h}_{\text{R}} \in \mathbb{R}_{\geq 0}^{N_T}$} \\ 
\cline{1-2} Jammer-Receiver's channel & $\overline{\textbf{h}}_{\text{R}} \in \mathbb{R}_{\geq 0}^{N_T - 1}$ &  \\
% \hline Alice-Eve' channel & $h_{\text{E}} \in \mathbb{R}_{\geq 0}$ & \multirow{2}{*}{$\mathbf{h}_{\text{E}} \in \mathbb{R}_{\geq 0}^{N_T}$} \\ 
% \cline{1-2} Jammer-Eve' channel & $\overline{\mathbf{h}}_{\text{E}} \in \mathbb{R}_{\geq 0}^{N_T - 1}$ &  \\
\hline Precoder of information symbol & $v\in \mathbb{R}$ & $\mathbf{v}\in \mathbb{R}^{N_T}$\\
\hline Precoder of AN & $\mathbf{w}\in \mathbb{R}^{N_T - 1}$ & $\mathbf{w}\in \mathbb{R}^{N_T}$\\
\hline
\end{tabular}}
\end{table}
With this, the LED's drive current vector in the case of unknown Eves' CSI and known Eves' CSI are given by $\mathbf{x} = \mathbf{v}d + \mathbf{w}z +I_{\text{DC}}$ and $\mathbf{x} = \mathbf{v}d + \sum_{k = 1}^K\mathbf{w}_kz_k + I_{\text{DC}}$, respectively. Then, constraints on $\mathbf{v}$, $\mathbf{w}$, and the received signals at the receiver follow the same analyses presented for the selective AN-aided SISO scheme.  
% one can derive the received signals at the receiver following the same as those described in \eqref{channelModel} and \eqref{receivedSignalKnownEve}. Furthermore, the LED's modulating current is a combination of the information-bearing signal and the AN, which is given by
% \begin{align}
%     \mathbf{x} = \mathbf{v}d + \mathbf{w}z.
% \end{align}
% The above expression gives rise to the following constraint on $\mathbf{v}$ and $\mathbf{w}$
% \begin{align}
% \left|[\mathbf{v}]_n\right| + \left|[\mathbf{w}]_n\right| \leq \Delta_{\text{DC}}, ~~~ \forall n = 1, 2, \hdots, N_T.
% \end{align}
\subsection{Energy Consumption and Energy Efficiency}
\subsubsection{Unknown Eves' CSI}
Since Alice and jammers have no knowledge about Eves, we consider a worst-case design based on an achievable rate of Bob's channel. It is observed that the channel in \eqref{channelModel} is subject to a Gaussian channel with amplitude-constrained input and interference. For such a channel, a lower bound on its channel capacity is given by \cite[Eq.~(4)]{Pham2019}, 
\begin{align}
    C_{\text{R}, l}(v, \mathbf{w}) = \frac{1}{2}\log_2\left(\frac{2\left(\left(h_{\text{R}}v\right)^2 + \left(\overline{\mathbf{h}}^T_{\text{R}}\mathbf{w}\right)^2\right) + \pi e\overline{\sigma}^2_{\text{R}}}{\pi e \left(\frac{1}{3}{\left(\overline{\mathbf{h}}^T_{\text{R}}\mathbf{w}\right)^2} + \overline{\sigma}^2_{\text{R}}\right)}\right).
    \label{achievableRateBob}
\end{align}
One can also derive an upper bound on the channel capacity based on \cite[Eq. (7)]{Pham2019}. However, the resulting design problem is not particularly different. For the sake of conciseness, we thus consider designs with respect to the lower bound capacity (or lower bound secrecy capacity as in the case of known Eves' CSI). 

In the considered system, energy is consumed for illumination, generation of the information-bearing signal and AN, and circuitry operations. In general, while the powers for illumination ($P_{\text{LEDs}}$) and circuitry operations ($P_{\text{circuit}}$) can be assumed to be fixed (i.e., the dimming level is unchanged), the power level by the information-bearing signal ($P_{\text{data}}$) and AN ($P_{\text{AN}}$) needs to be optimized with respect to the achievable secrecy rate. The total power can then be written as
\begin{align}
    P_{\text{total}} & =  P_{\text{circuit}} + P_{\text{LEDs}} + P_{\text{data}} + P_{\text{AN}} 
     = P_{\text{circuit}} + \sum_{N_T} U_{\text{LEDs}}I_{\text{DC}} + \zeta\left(v^2 + \lVert\mathbf{w}\rVert^2\right),
\end{align}
where $U_{\text{LEDs}}$ is the LEDs' forward voltage and $\zeta = \frac{1}{3}R_{\text{AC}}$ with $R_{\text{AC}}$ being the resistance of the AC circuit. Given the achievable rate in \eqref{achievableRateBob}, the EE of the system with respect to Bob's channel is expressed by 
\begin{align}
    \Phi_{\text{B}}(v, \mathbf{w}) = \frac{C_{\text{B}, l}(v, \mathbf{w})}{P_{\text{circuit}} + \sum_{N_T} U_{\text{LEDs}}I_{\text{DC}} + \zeta\left(v^2 + \lVert\mathbf{w}\rVert^2\right)}.
    \label{SEE}
\end{align}

\subsubsection{Known Eves' CSI}
In this scenario, due to the availability of Eves' CSI, achievable secrecy rates can be derived to calculate the SEE. In particular, a lower bound  on the secrecy capacity of the wiretap channel comprising Bob and the $k$th Eve is written by
\begin{align}
    C^k_s\left(v, \mathbf{W}\right) = C_{\text{B}, l}\left(v, \mathbf{W}\right) - C^k_{\text{E},u}\left(v, \mathbf{W}\right),
\end{align}
where $\mathbf{W} = \begin{bmatrix}
    \mathbf{w}_1 & \mathbf{w}_2 & \cdots & \mathbf{w}_K
\end{bmatrix}$ and $C^k_{\text{E}, u}\left(v, \mathbf{W}\right)$ is an upper bound on the capacity of the $k$-th Eve's channel, which can be derived from \cite[Eq.~(7)]{Pham2019}. Using \eqref{achievableRateBob} gives the following expression for $C^k_s(v, \mathbf{W})$
\begin{align}
    C^k_{\text{s}}(v, \mathbf{W}) & =  \frac{1}{2}\log_2\left(\frac{2\left(\left(h_{\text{B}}v\right)^2 + \left\lVert\overline{\mathbf{h}}^T_{\text{B}}\mathbf{W} \right\rVert^2\right) + \pi e \overline{\sigma}^2_{\text{B}}}{\pi e\left(\frac{1}{3}\left\lVert\overline{\mathbf{h}}^T_{\text{B}}\mathbf{W}\right\rVert^2 + \overline{\sigma}^2_{\text{B}}\right)}\right) \nonumber \\ &- \frac{1}{2}\log_2\left(\frac{\pi e \left(\frac{1}{3}\left(\left(h_{\text{E}, k}v\right)^2 + \left\lVert\overline{\mathbf{h}}^T_{\text{E}, k}\mathbf{W}\right\rVert^2\right) + \overline{\sigma}^2_{\text{E}, k}\right)}{2\left\lVert\overline{\mathbf{h}}^T_{\text{E}, k}\mathbf{W}\right\rVert^2 + \pi e \overline{\sigma}^2_{\text{E}, k}}\right),
    \label{lowerboundSR}
\end{align}
with $h_{\text{E}, k}$ and $\overline{\mathbf{h}}_{\text{E}, k}$ being channel coefficients between Alice, jammer and the $k$-th Eve, respectively. Moreover, the constraint in \eqref{constraint_AN_precoder} can rewritten in terms of $\mathbf{W}$ as 
\begin{align}
    \left\lVert\mathbf{W}\right\rVert_{\infty} \leq \Delta_{\text{DC}}.
\end{align}
Using the same assumptions in the case of unknown Eves' CSI, the total power consumption is given by 
\begin{align}
    P_{\text{total}} & =  P_{\text{circuit}} + P_{\text{LEDs}} + P_{\text{data}} + P_{\text{AN}} 
     = P_{\text{circuit}} + \sum_{N_T} U_{\text{LEDs}}I_{\text{DC}} + \zeta\left(v^2 + \text{tr}\left(\mathbf{W}\mathbf{W}^T\right)\right).
\end{align}

Note that to ensure the confidentiality of the transmitted signal against all eavesdroppers, the information rate must be no higher than the lowest achievable secrecy rate among $K$ values of $C^k_s(v, \mathbf{W})$. Therefore, we are interested in an AN design that maximizes the SEE corresponding to the minimum achievable secrecy rate, which is defined by
\begin{align}
    \widetilde{\Phi}^{\text{min}}_s(v, \mathbf{W}) = \frac{\underset{k}{\text{min}}\hspace{2mm}C^k_s(v, \mathbf{W})}{P_{\text{circuit}} + \sum_{N_T}U_{\text{LEDs}}I_{\text{DC}} + \zeta\left(v^2 + \text{tr}\left(\mathbf{W}\mathbf{W}^T\right)\right)}.
    \label{lowerboundSR_k}
\end{align}
%=================================================
\section{Secrecy Energy Efficiency Maximization}
The availability of Eves's CSI at Alice and the jammers is critical to our design approaches. In the following, we, hence, separately consider two scenarios: unknown and known Eves' CSI. The designs are again presented for the selective SISO transmission scheme, mentioning the MISO transmission scheme when needed. 
\subsection{Unknown Eves' CSI}
In most practical cases, Eves are passive eavesdroppers who do not feed their CSI back to the transmitters, leading to the unavailability of Eves' CSI on the transmitting side. 
%Note that the SEE in \eqref{SEE} is a function of the achievable secrecy rate $C_{\text{s}}(v, \mathbf{w})$, which can only be derived and utilized when the instantaneous Eves' CSI is available. 
Without the knowledge of Eves' CSI, direct optimization of the SEE may not be possible. Instead, an indirect design approach is to maximize the EE of Bob's channel. The obtained solution $(v, \mathbf{w})$ is then used to calculate the SEE.  Due to the proportionality between the capacity of Bob's channel and the secrecy capacity, this design approach possibly results in good SEE performance as well. 
% Similar to \eqref{SEE}, the EE of Bob's channel is defined by
% \begin{align}
%     \Phi_{\text{B}}(v, \mathbf{w}) = \frac{C_{\text{B}}(v, \mathbf{w})}{P_{\text{circuit}} + \sum_{N_T} U_{\text{LEDs}}I_{\text{DC}} + \zeta\left(v^2 + \lVert\mathbf{w}\rVert^2\right)}.
%     \label{BobEE}
% \end{align}

With constraints on the precoders $v$ and $\mathbf{w}$ given in \eqref{precoder-constraint-1} and \eqref{precoder-constraint-2}, a naive way to formulate an EE maximization problem for Bob's channel is as follows%\footnote{In the case of the full transmission scheme, $\mathcal{P}\textbf{1}$ is a maximization problem over $\mathbf{v}$ and $\mathbf{w}$ with a constraint being $\left\lVert\mathbf{v} + \mathbf{w}\right\rVert_\infty \leq \Delta_{\text{DC}}$.} 
\begin{subequations}
\begin{alignat}{2}
\bf{\mathcal{P}1} \hspace{5mm} & \underset{v, \mathbf{w}}{\text{maximize}}  & \hspace{2mm}  & \Phi_{\text{B}}(v, \mathbf{w}) \label{P1:obj_func}\\
& \text{subject to} &		& \nonumber \\
&  & &  |v| \leq \Delta_{\text{DC}},  \label{P1:constraint1}\\ 
&  & & \lVert\mathbf{w}\rVert_{\infty} \leq \Delta_{\text{DC}}. \label{P1:constraint2} 
\end{alignat}
\end{subequations}
%However, there are two issues with the above design problem. Firstly, we have the following.
\begin{Proposition}
The optimal AN solution to $\mathcal{P}\mathbf{{1}}$ is $\mathbf{w}^* = \mathbf{0}$. 
\end{Proposition}
\begin{proof}
An intuitive proof is that using AN does not improve the achievable rate of Bob's channel while increasing energy consumption. Thus, an AN-aided transmission reduces the EE, which shows that the no-AN transmission (i.e., $\mathbf{w} = \mathbf{0}$) is optimal. A rigorous proof is given in Appendix A. 
\end{proof}
\textbf{Proposition 1} shows that in order to achieve the maximal $\Phi_{\text{B}}(v, \mathbf{w})$, jammers should be inactive. However, this may not lead to an improved SEE because without being interfered with by the AN, Eves can enjoy high-quality communication channels (thus, lowering the achievable secrecy rate). To enable the use of AN, one can set a constraint to ensure that a certain amount of power is allocated to generating the AN signal, for example, $\lVert\mathbf{w}\rVert^2_2 \geq P_{\text{th}}$ where $P_{\text{th}} > 0$ is the minimal allocated power threshold. The choice of $P_{\text{th}}$ is critical as due to \eqref{P1:constraint2}, this power allocation renders the design problem infeasible if $P_{\text{th}} > (N_T - 1)\Delta_{\text{DC}}^2$. Moreover, although increasing $P_{\text{th}}$ decreases Eves' SINRs, it potentially degrades the quality of Bob's channel as well and increases the total consumption power. Consequently, the overall SEE might be reduced. To prevent too much power from being allocated to the AN and guarantee that the quality of Bob's channel is above a desirable limit, a constraint on Bob's SINR is taken into consideration. Specifically, we study a modification 
for $\bf{\mathcal{P}1}$ as given below   
\begin{subequations}
\begin{alignat}{2}
\bf{\mathcal{P}2} \hspace{5mm} & \underset{v, \mathbf{w}}{\text{maximize}} &   \hspace{2mm}  & \Phi(v, \mathbf{w}) \label{P2:obj_func}\\
& \text{subject to} &	& \nonumber \\
&  & &  \frac{\left(h_{\text{B}}v\right)^2}{\left(\overline{\mathbf{h}}^T_{\text{B}}\mathbf{w}\right)^2 +   \overline{\sigma}^2_{\text{B}}} \geq \delta_{\text{B}}, \label{P2:constraint1} \\
&  & &   \lVert\mathbf{w}\rVert_2^2 \geq P_{\text{th}}, \label{P2:constraint2} \\
% &  & &   |v|  \leq \Delta_{\text{DC}},  \label{P2:constraint3}\\ 
% & & &    \lVert\mathbf{w}\rVert_{\infty} \leq \Delta_{\text{DC}}. \label{P2:constraint4} 
& & & \eqref{P1:constraint1},~\eqref{P1:constraint2}. \nonumber 
\end{alignat}
\end{subequations}
Here, $\delta_{\text{B}}$ is the minimum SINR required for Bob's channel. 

One should emphasize that the introduction of constraints in \eqref{P2:constraint1} and \eqref{P2:constraint2} may render $\bf{\mathcal{P}2}$ infeasible. Indeed, it can be easily verified that the problem is infeasible when $\delta_{\text{B}} > \frac{\left(h_{\text{B}}\Delta_{\text{DC}}\right)^2}{\overline{\sigma}^2_{\text{B}}}$.
Therefore, we assume that $\delta_{\text{B}}$ and $P_{\text{th}}$ are explicitly chosen so that $\bf{\mathcal{P}2}$ is feasible. It is seen that $\bf{\mathcal{P}2}$ is a fractional non-convex optimization problem. Hence, we employ the Dinkelbach algorithm, which is efficient in solving fractional programming to tackle the problem. The principle of the Dinkelbach algorithm is to get rid of the fractional objective function by introducing an equivalent non-fractional parametric objective function, which can be handled more conveniently. In particular to our problem, let $D(v, \mathbf{w}) = P_{\text{circuit}} + \sum_{N_T} U_{\text{LEDs}}I_{\text{DC}} + \zeta\left(v^2 + \lVert\mathbf{w}\rVert_2^2\right)$ and $\lambda \geq 0$ be the value of the objective function. The algorithm then iteratively solves the following parametric problem 
\begin{subequations}
\begin{alignat}{2}
\bf{\mathcal{P}3}(\lambda) \hspace{5mm} & \underset{v, \mathbf{w}}{\text{maximize}} &  \hspace{2mm} & C_{\text{B}}(v, \mathbf{w}) - \lambda D(v, \mathbf{w}) \label{P3:obj_func}\\
& \text{subject to} &		& \nonumber \\
&  &  & \eqref{P2:constraint1},~\eqref{P2:constraint2}, ~\eqref{P1:constraint1},~\eqref{P1:constraint2}. \nonumber 
\end{alignat}
\label{P3-parametric}
\end{subequations}
In each iteration, say the $i$th iteration, a new $\lambda^{(i)} = \frac{C_{\text{B}}\left(v^{{(i)}^*}, \mathbf{w}^{{(i)}^*}\right)}{D\left(v^{{(i)}^*}, \mathbf{w}^{{(i)}^*}\right)}$ is updated, where $v^{{(i)^*}}$ and $\mathbf{w}^{{(i)^*}}$
are the iteration's optimal solutions to $v$ and $\mathbf{w}$, respectively. The algorithm terminates when $C_{\text{B}}\left(v^{{(i)^*}}, \mathbf{w}^{{(i)^*}}\right) - \lambda^{(i)}D\left(v^{{(i)^*}}, \mathbf{w}^{{(i)^*}}\right)$ converges to a predefined sufficiently low value. The algorithm is summarized as follows.  
\begin{algorithm2e}[ht]
\SetAlgoLined % activate/deactivate number line
\caption{Dinkelbach-type algorithm}
Choose the maximum number of iterations $L_{\text{Din}}$ and the error tolerance $\epsilon_1 > 0$. \\
Set $i \leftarrow 1$, $\lambda^{(0)} \leftarrow 0$. \\
\While{\rm{convergence} = \textbf{False} and $i \leq L_{\text{Din}}$}{
Given $\lambda^{(i-1)}$ from the previous iteration, solve ${\bf{\mathcal{P}3}}\left(\lambda^{(i-1)}\right)$ for $v^{{(i)^*}}$ and $\mathbf{w}^{{(i)^*}}$. \\
Update $\lambda^{(i)} \leftarrow \frac{C_{\text{B}}\left(v^{{(i)^*}}, \mathbf{w}^{{(i)^*}}\right)}{D\left(v^{{(i)^*}}, \mathbf{w}^{{(i)^*}}\right)}$. \\
\eIf{$C_{\text{B}}\left(v^{{(i)^*}}, \mathbf{w}^{{(i)^*}}\right)- \lambda^{(i)}D\left(v^{{(i)^*}}, \mathbf{w}^{{(i)^*}}\right) \leq \epsilon_{1}$}{convergence $\leftarrow$ \textbf{True}\\
$v^{*} \leftarrow v^{{(i)^*}}$,~ 
$\mathbf{w}^{*} \leftarrow \mathbf{w}^{{(i)^*}}$ \\
}{
convergence $\leftarrow \textbf{False}$}
$i \leftarrow i + 1$
}
\label{DinkebackAlgo1}
\end{algorithm2e}
\subsubsection{General AN-aided SISO}
In each iteration of \textbf{Algorithm \ref{DinkebackAlgo1}}, it is required to solve $\bf{\mathcal{P}3}(\lambda)$, which  is not a convex optimization problem due to the non-concavity of the objective function and non-convexity of \eqref{P2:constraint1} and \eqref{P2:constraint2}. 
Notice that \eqref{P2:constraint1} can be convexified via an observation that there exists an optimal solution $v^*$ satisfying that $h_{\text{B}}v^* \geq 0$. Indeed, it is obvious that if there is any optimal solution $v^{\dagger}$ to $\bf{\mathcal{P}3}(\lambda)$ that $h_{\text{B}}v^{\dagger} \leq 0$, then $v^{*} = -v^{\dagger}$ is also feasible and offers the same objective value. %\footnote{Following the same argument, in the case of the AN-aided MISO transmission scheme, there exists an optimal solution $\mathbf{v}^*$ such that $\mathbf{h}^T_{\text{B}}\mathbf{v}^* \geq 0$.}. 
Hence, $v^*$ is also optimal. As a result, \eqref{P2:constraint1} is equivalently rewritten as 
\begin{align}
\frac{1}{\sqrt{\delta_{\text{B}}}}h_{\text{B}}v \geq \sqrt{\left(\overline{\mathbf{h}}^T_{\text{B}}\mathbf{w}\right)^2 +   \overline{\sigma}^2_{\text{B}}},
\end{align}
which is convex. The non-concavity of the objective function is due to the fact that $C_{\text{B}}(v, \mathbf{w})$ is not concave. To overcome this problem, one can make use of the following variable transformations 
\begin{align}
&c_{\text{B}, 1} = \log_2\left(2\left(\left(h_{\text{B}}v\right)^2 + \left(\overline{\mathbf{h}}^T_{\text{B}}\mathbf{w}\right)^2\right) + \pi e\overline{\sigma}^2_{\text{B}}\right), \label{cB1}\\
&c_{\text{B}, 2} = \log_2\left(\pi e \left(\frac{1}{3}{\left(\overline{\mathbf{h}}^T_{\text{B}}\mathbf{w}\right)^2} + \overline{\sigma}^2_{\text{B}}\right)\right), \label{cB2}\\
&p_{\text{B}, 1} = \left(h_{\text{B}}v\right)^2 + \left(\overline{\mathbf{h}}^T_{\text{B}}\mathbf{w}\right)^2, \label{pB1}\\
&p_{\text{B}, 2} = \left(\overline{\mathbf{h}}^T_{\text{B}}\mathbf{w}\right)^2 \label{pB2}. 
\end{align}
Accordingly, we reformulate $\bf{\mathcal{P}3}(\lambda)$ as 
\begin{subequations}
\begin{alignat}{2}
{\bf{\mathcal{P}4}}(\lambda) \hspace{5mm} & \underset{v, \mathbf{w}, c_{\text{B}, 1},  c_{\text{B}, 2}, p_{\text{B}, 1}, p_{\text{B}, 2}}{\text{maximize}}  & \hspace{2mm} & \frac{1}{2}\left(c_{\text{B}, 1} - c_{\text{B}, 2}\right) - \lambda D(v, \mathbf{w}) \label{P4:obj_func}\\
& \text{subject to} &		& \nonumber \\
& & &c_{\text{B}, 1} \leq \log_2\left(2p_{\text{B}, 1} + \pi e \overline{\sigma}^2_{\text{B}}\right), \label{P4:constraint1} \\
& & &p_{\text{B}, 1} \leq \left(h_{\text{B}}v\right)^2 + \left(\overline{\mathbf{h}}^T_{\text{B}}\mathbf{w}\right)^2, \label{P4:constraint2} \\
& & &c_{\text{B}, 2} \geq \log_2\left(\pi e \left(\frac{1}{3}{p_{\text{B}, 2}} + \overline{\sigma}^2_{\text{B}}\right)\right), \label{P4:constraint3} \\
& & & p_{\text{B}, 2} \geq \left(\overline{\mathbf{h}}^T_{\text{B}}\mathbf{w}\right)^2, \label{P4:constraint4} \\
& &  &  \frac{1}{\sqrt{\delta_{\text{B}}}}h_{\text{B}}v \geq \sqrt{\left(\overline{\mathbf{h}}^T_{\text{B}}\mathbf{w}\right)^2 +   \overline{\sigma}^2_{\text{B}}},
 \label{P4:constraint5} \\
&  &  & \eqref{P2:constraint2}, \eqref{P1:constraint1}, \eqref{P1:constraint2}. \nonumber 
\end{alignat}
\end{subequations}
Note that without lost of optimality, the equalities in \eqref{cB1}-\eqref{pB2} are replaced by their corresponding inequalities in \eqref{P4:constraint1}-\eqref{P4:constraint4} of $\mathbf{\mathcal{P}4}(\lambda)$. Indeed, since the objective in \eqref{P4:obj_func} is monotically increasing with $c_{\text{B}, 1}$, at the opitmal solution, $c_{\text{B}, 1}$ must attain its maximum value, which implies that \eqref{P4:constraint1} and \eqref{P4:constraint2} hold at equality. The same argument can be used to validate \eqref{P4:constraint3} and \eqref{P4:constraint4}.
Now, all constraints except \eqref{P4:constraint2}, \eqref{P4:constraint3}, and \eqref{P2:constraint2} are convex. Observe that one side in each of these inequality constraints is either constant or linear. This enables the use of Taylor expansion to approximately linearize the constraints. Then, the CCP is used to successively solve the problem. Specifically, at the $j$th iteration of the CCP, the following first-order Taylor approximations are used
\begin{align}
\left(h_{\text{B}}v^{(j)}\right)^2 + \left(\overline{\mathbf{h}}^T_{\text{B}}\mathbf{w}\right)^2 \geq  & \left(h_{\text{B}}v^{(j-1)}\right)^2 + 2h^2_{\text{B}}v^{(j-1)}\left(v^{(j)}  - v^{(j-1)}\right)
 + \left(\overline{\mathbf{h}}^T_{\text{B}}\mathbf{w}^{(j-1)}\right)^2  \nonumber \\ & +  2\left[\mathbf{w}^{(j-1)}\right]^{T}\overline{\mathbf{h}}_{\text{B}}\overline{\mathbf{h}}^T_{\text{B}}\left(\mathbf{w}^{(j)} - \mathbf{w}^{(j-1)}\right),
\label{TaylorApprox1}
\end{align}
\begin{align}
\log_2\left(\pi e \left(\frac{1}{3}{p_{\text{B,2}}} + \overline{\sigma}^2_{\text{B}}\right)\right)  \leq \log_2\left(\!\pi e \left(\frac{1}{3}{p^{(j-1)}_{\text{B}, 2}} + \overline{\sigma}^2_{\text{B}}\right)\!\right)  + \frac{p^{(j)}_{\text{B}, 2} - p^{(j-1)}_{\text{B}, 2}}{\ln(2)\left(\!\pi e \left(\frac{1}{3}{p^{(j-1)}_{\text{B}, 2}} + \overline{\sigma}^2_{\text{B}}\right)\!\right)},
\label{TaylorApprox2}
\end{align}
\begin{align}
\left\lVert\mathbf{w}^{(j)}\right\rVert_2^2 \geq \left\lVert\mathbf{w}^{(j-1)}\right\rVert_2^2 + 2\left[\mathbf{w}^{(j-1)}\right]^T\left(\mathbf{w}^{(j)} - \mathbf{w}^{(j-1)}\right),
\label{TaylorApprox3}
\end{align}
where $v^{(j-1)}$, $\mathbf{w}^{(j-1)}$, and $p^{(j-1)}_{\text{B, 2}}$ are the solutions to $v$, $\mathbf{w}$, and $p_{\text{B, 2}}$ in the $(j-1)$th iteration, respectively. Replacing \eqref{P4:constraint2}, \eqref{P4:constraint3}, and \eqref{P2:constraint2} by their corresponding constraints using \eqref{TaylorApprox1}, \eqref{TaylorApprox2}, and \eqref{TaylorApprox3} results in 
\begin{subequations}
\begin{alignat}{2}
{\bf{\mathcal{P}5}}^{(j)}(\lambda) \hspace{2mm} & \underset{\substack{v^{(j)}, \mathbf{w}^{(j)}, c_{\text{B}, 1}, \\ c_{\text{B}, 2}, p_{\text{B}, 1}, p^{(j)}_{\text{B}, 2}}}{\text{maximize}}  & \hspace{2mm} & \frac{1}{2}\left(c_{\text{B}, 1} - c_{\text{B}, 2}\right) - \lambda D\left(v^{(j)}, \mathbf{w}^{(j)}\right) \label{P5:obj_func}\\
& \text{subject to} &		& \nonumber \\
& & &\left(h_{\text{B}}v^{(j-1)}\right)^2 +  2h_{\text{B}}\left(v^{(j)} - v^{(j-1)}\right) \nonumber \\ & & & + \left(\overline{\mathbf{h}}^T_{\text{B}}\mathbf{w}^{(j-1)}\right)   + 2\left[\mathbf{w}^{(j-1)}\right]^{T}\overline{\mathbf{h}}_{\text{B}}\overline{\mathbf{h}}^T_{\text{B}}\left(\mathbf{w}^{(j)} - \mathbf{w}^{(j-1)}\right)  \geq p_{\text{B, 1}}, \label{P5:constraint1} \\
& & & c_{\text{B, 2}} \geq \log_2\left(\!\pi e \left(\frac{1}{3}{p^{(j-1)}_{\text{B}, 2}} + \overline{\sigma}^2_{\text{B}}\right)\!\right) + \frac{p^{(j)}_{\text{B}, 2} - p^{(j-1)}_{\text{B}, 2}}{\ln(2)\left(\!\pi e \left(\frac{1}{3}{p^{(j-1)}_{\text{B}, 2}} + \overline{\sigma}^2_{\text{B}}\right)\!\right)}, \label{P5:constraint2} \\
& & &\left\lVert\mathbf{w}^{(j-1)}\right\rVert_2^2 + 2\left[\mathbf{w}^{(j-1)}\right]^T\!\!\left(\mathbf{w}^{(j)} - \mathbf{w}^{(j-1)}\right)  \geq P_{\text{th}}, \label{P5:constraint3} \\
&  &  & \eqref{P4:constraint1},~\eqref{P4:constraint4},~ \eqref{P4:constraint5},~\eqref{P1:constraint1},~\eqref{P1:constraint2}, \nonumber 
\end{alignat}
\end{subequations}
which is a convex optimization problem and thus can be solved efficiently using available software packages, e.g., CVX \cite{cvx}. Then a sub-optimal solution to $\bf{\mathcal{P}4}(\lambda)$ can be found via solving a sequence of ${\bf{\mathcal{P}5}}^{(j)}$ until a convergence criterion is satisfied. Note that due to the tighter constraints used in \eqref{P5:constraint1}-\eqref{P5:constraint2}, it is guaranteed that a solution to ${\bf{\mathcal{P}5}}^{(j)}(\lambda)$ is feasible to  ${\bf{\mathcal{P}5}}^{(j + 1)}(\lambda)$.
In summary, an CCP algorithm to solve ${\bf{\mathcal{P}4}}(\lambda)$ is described in \textbf{Algorithm 2}. 
\begin{algorithm2e}
\caption{CCP-type algorithm for solving ${\bf{\mathcal{P}4}}(\lambda)$}
Choose the maximum number of iterations $L_{\text{CCP}}$ and the error tolerance $\epsilon_2$. \\
Choose the starting points $v^{(0)}$,  $\mathbf{w}^{(0)}$, and $p^{(0)}_{\text{B, 2}}$ satisfying that ${\bf{\mathcal{P}5}}^{(0)}(\lambda)$ is feasible. \\ 
Set $j \leftarrow 1$. \\

\While{\rm{convergence} = \textbf{True} and $j < L_{\text{CCP}}$}{
Given $v^{(j-1)}$,  $\mathbf{w}^{(j-1)}$, and $p^{(j-1)}_{\text{B, 2}}$ obtained from the previous iteration, solve ${\bf{\mathcal{P}5}}^{(j)}(\lambda)$. \\
\eIf{$\frac{\left|v^{(j)} - v^{(j-1)}\right|}{v^{(j)}} \leq \epsilon_{2}$ {\rm{and}} $\frac{\left\lVert \mathbf{w}^{(i)} - \mathbf{w}^{(i-1)}\right\rVert_2}{\left\lVert\mathbf{w}^{(i)}\right\rVert_2} \leq \epsilon_{2}$ {\rm{and}} $\frac{\left|p^{(j)}_{\text{B, 2}} - p^{(j-1)}_{\text{B, 2}}\right|}{p^{(j)}_{\text{B}, 2}} \leq \epsilon_{2}$
}{convergence $\leftarrow$ \textbf{True}.\\
$v^{*} \leftarrow v^{(j)}$,~
$\mathbf{w}^{*} \leftarrow \mathbf{w}^{(j)}$,~ 
$p^{*}_{\text{B}, 2} \leftarrow p^{(j)}_{\text{B}, 2}$. \\
}{convergence $\leftarrow$ \textbf{False}.
} $j \leftarrow j + 1$. \\ 
}
\label{Alg2}
\end{algorithm2e}

%The obtained solution to $\mathcal{P}\bf{2}$ is then used to calculate the SEE give in 
\subsubsection{ZF AN-aided SISO} The presented general AN design requires two iterative procedures, which might incur a high computational complexity. In order to simplify the design, we present in this section a ZF design approach where the AN is constructed to lie on the null-space of $\overline{\mathbf{h}}_{\text{B}}^T$. This ensures that the generated AN does cause any interference on Bob's channel. This, however, also leads to a smaller search space of $\mathbf{w}$, possibly resulting in a reduced SEE performance. To this end, let $\widetilde{\mathbf{w}}$ be an orthogonal basic for the null-space of $\overline{\mathbf{h}}_\text{B}^T$. The AN precoder is then given by $\mathbf{w} = \sqrt{\phi}\widetilde{\mathbf{w}}$, where $\phi > 0$ is a scaling factor, which controls the AN power. Under this null-space strategy, the EE of Bob's channel is simplified to 
\begin{align}
    \Phi_{\text{B}}(v, \phi) = \frac{\log_2\left(1 + \frac{2\left(h_{\text{B}}v\right)^2}{\pi e \overline{\sigma}^2_{\text{B}}}\right)}{2\left(P_{\text{circuit}} + \sum_{N_T} U_{\text{LEDs}}I_{\text{DC}} + \zeta\left(v^2 + \phi\right)\right)},
    \label{BobEE-NullAN}
\end{align}

and the design problem becomes 
\begin{subequations}
\begin{alignat}{2}
\bf{\mathcal{P}6} \hspace{5mm} & \underset{v, \phi}{\text{maximize}} &  \hspace{2mm} &  \Phi_{\text{B}}(v, \phi) \label{P6:obj_func}\\
& \text{subject to} &		& \nonumber \\
& &  &  \frac{\left(h_{\text{B}}v\right)^2}{  \overline{\sigma}^2_{\text{B}}} \geq \delta_{\text{B}}, \label{P6:constraint1} \\
& &   & \phi  \geq P_{\text{th}}, \label{P6:constraint2} \\
&  &  & |v| \leq \Delta_{\text{DC}},  \label{P6:constraint3}\\ 
& &  &  \sqrt{\phi}  \leq \frac{\Delta_{\text{DC}}}{\lVert\widetilde{\mathbf{w}}\rVert_{\infty}}. \label{P6:constraint4} 
\end{alignat}
\end{subequations}
Note that \eqref{P6:constraint2} holds because $\left\lVert\widetilde{\mathbf{w}}\right\rVert_2 = 1$. Firstly, we can see that $\bf{\mathcal{P}6}$ is feasible if and only if the two conditions $\frac{\Delta_{\text{DC}}^2}{\left\lVert\widetilde{\mathbf{w}}\right\rVert^2_{\infty}} \geq {P_{\text{th}}}$ and $\Delta_{\text{DC}}^2 \geq \frac{\delta_{\text{B}}\overline{\sigma}^2_{\text{B}}}{h^2_{\text{B}}}$ are simultaneously satisfied. Assume that these conditions hold, because \eqref{BobEE-NullAN} is monotonically decreasing with $\phi$, the optimal solution to $\phi$ is obviously $\phi^* = P_{\text{th}}$. With this and by making change of variable $V = v^2$, $\bf{\mathcal{P}6}$ reduces to 
\begin{subequations}
\begin{alignat}{2}
\bf{\mathcal{P}7}(\phi^*) \hspace{5mm} & \underset{V \geq 0}{\text{maximize}}  & \hspace{2mm}  & \Phi_{\text{B}}(V, \phi^*) \label{P7:obj_func}\\
& \text{subject to} &		& \nonumber \\
& &   & \frac{h^2_{\text{B}}V}{  \overline{\sigma}^2_{\text{B}}} \geq \delta_{\text{B}}, \label{P7:constraint1} \\
& &  &  V \leq \Delta_{\text{DC}}^2,   \label{P7:constraint2}
\end{alignat}
\end{subequations}
 where $\Phi_{\text{B}}(V, \phi^*) = \frac{\log_2\left(1 + \frac{2h^2_{\text{B}}V}{\pi e \overline{\sigma}^2_{\text{B}}}\right)}{2\left(P_{\text{circuit}} + \sum_{N_T} U_{\text{LEDs}}I_{\text{DC}} + \zeta\left(V + \phi^*\right)\right)}$.
\begin{Proposition}
Assume that $\bf{\mathcal{P}7}(\phi^*)$ is feasible, a global solution $V^*$  exists and is unique. 
\end{Proposition}
\begin{proof}
The proof is given in Appendix B. 
\end{proof}
The proof of  $\textbf{Proposition 2}$ shows that  
there is no need to utilize the Dinkelbach algorithm and the CCP to solve $\mathcal{P}\mathbf{6}$. In the worst-case scenario, solving $\mathcal{P}\mathbf{6}$ only requires a bisection procedure, which generally results in a considerably reduced computational complexity. Note that the same solving approach presented in the proof of the proposition may not be possible for the ZF AN-aided MISO scheme because of the vector form of the precoder $\mathbf{v}$. Instead, a combination of the Dinkelbach algorithm and the CCP is required to solve the design problem.  Nonetheless, the complexity can still be reduced due to the smaller search space of $\mathbf{w}$ when the ZF constraint is applied. 
\subsection{Known Eves' CSI}
Recall that when Eves' CSI is known at the transmitter, we consider a direct design to maximize the minimum SEE.  
Following the expression in \eqref{lowerboundSR_k}, the AN design problem, in this case, is formulated as 
\begin{subequations}
\begin{alignat}{2}
{\bf{\mathcal{P}8}} \hspace{5mm} & \underset{v, \mathbf{W}}{\text{maximize}}   & \hspace{2mm} & \widetilde{\Phi}^{\text{min}}_s(v, \mathbf{W}) \label{P8:obj_func}\\
& \text{subject to} &		& \nonumber \\
%& & &  \underset{k}{\text{min}}\hspace{2mm}C^k_s(v, \mathbf{W}) \geq \xi_{\text{th}}, \label{P8:constraint1} \\
& & &|v|  \leq \Delta_{\text{DC}}, \label{P8:constraint1} \\ 
& & &\left\lVert\mathbf{W}\right\rVert_{\infty} \leq \Delta_{\text{DC}}.
\label{P8:constraint2}
\end{alignat}
\label{P8-parametric}
\end{subequations}
It is seen that $\mathcal{P}\textbf{8}$ is a nonconvex maximin optimization problem. A standard technique to handle the problem is to introduce a slack variable $t = \widetilde{\Phi}^{\text{min}}_s(v, \mathbf{W})$ and then transform  $\mathcal{P}\textbf{8}$ to 
\begin{subequations}
\begin{alignat}{2}
{\bf{\mathcal{P}9}} \hspace{5mm} & \underset{v, \mathbf{W}, t}{\text{maximize}}   \hspace{2mm}  t \label{P9:obj_func} \\
& \text{subject to} &	  & \nonumber \\
%& & & \widetilde{\Phi}_s(v, \mathbf{W}) \geq t & \forall k = 1, 2, ..., K,  \label{P9:constraint1} \\
&  C^k_s(v, \mathbf{W}) \geq t\left(P_{\text{circuit}} + \sum_{N_T} U_{\text{LEDs}}I_{\text{DC}} + \zeta\left(v^2 + \text{tr}\left(\mathbf{W}\mathbf{W}^T\right)\right)\!\!\right) ~~\forall k = 1, 2, ..., K, \label{P9:constraint1} \\
%& C_s(v, \mathbf{W}) \geq \xi_{\text{th}},  & \forall k = 1, 2, ..., K, \label{P9:constraint2} \\
%& |v|  \leq \Delta_{\text{DC}}, & \label{P9:constraint3} \\ 
%& \left\lVert\mathbf{W}\right\rVert_{\infty} \leq \Delta_{\text{DC}}. & \label{P9:constraint4}
& \eqref{P8:constraint1},~ \eqref{P8:constraint2}. \nonumber
\end{alignat}
\label{P9-parametric}
\end{subequations}
The use of the slack variable $t$ results in additional constraints $\widetilde{\Phi}_s(v, \mathbf{W}) \geq t, \hspace{2pt} \forall k = 1, 2, ..., K$,  which are equivalently represented in \eqref{P9:constraint1}, %Similarly, the constraint in \eqref{P8:constraint1} is rewritten in \eqref{P9:constraint2}. Note that both \eqref{P9:constraint1} and \eqref{P9:constraint2} are non-convex. While \eqref{P9:constraint2} can be handled using variable transformations and the CCP,
which is non-convex and difficult to handle due to the product between optimizing variables on the right-hand side. It is, therefore, challenging to directly solve ${\bf{\mathcal{P}9}}$. However, one can transform ${\bf{\mathcal{P}9}}$ into a feasibility problem by fixing the value of $t$. Then, a solution to the original problem can be found using the bisection method. Specifically, given a fixed value of $t$, 
\begin{subequations}
\begin{alignat}{2}
{{\mathcal{P}\bf{10}}(t)} \hspace{3mm} & \text{find}   \hspace{2mm}  v, \mathbf{W} \label{P10:obj_func} \\
& \text{subject to} &	  & \nonumber \\
%& & & \widetilde{\Phi}_s(v, \mathbf{W}) \geq t & \forall k = 1, 2, ..., K,  \label{P9:constraint1} \\
&  C^k_s(v, \mathbf{W}) \geq t\left(P_{\text{circuit}} + \sum_{N_T} U_{\text{LEDs}}I_{\text{DC}} + \zeta\left(v^2 + \text{tr}\left(\mathbf{W}\mathbf{W}^T\right)\right)\!\!\right)  ~~\forall k = 1, 2, ..., K, \label{P10:constraint1} \\
%& C_s(v, \mathbf{W}) \geq \xi_{\text{th}},  & \forall k = 1, 2, ..., K, \label{P10:constraint2} \\
& \eqref{P8:constraint1}, ~\eqref{P8:constraint2}. \nonumber 
\end{alignat}
\label{P10-parametric}
\end{subequations}
Similar to \eqref{cB1}-\eqref{pB2}, we make use of the following variable transformations
\begin{align}
&c_{\text{B}, 1} = \log_2\left(2\left(\left(h_{\text{B}}v\right)^2 + \left\lVert\overline{\mathbf{h}}^T_{\text{B}}\mathbf{W}\right\rVert^2\right) + \pi e\overline{\sigma}^2_{\text{B}}\right), \label{cB1_k} \\
&c_{\text{B}, 2} = \log_2\left(\pi e \left(\frac{1}{3}{\left\lVert\overline{\mathbf{h}}^T_{\text{B}}\mathbf{W}\right\rVert^2} + \overline{\sigma}^2_{\text{B}}\right)\right), \label{cB2_k}\\
&p_{\text{B}, 1} = \left(h_{\text{B}}v\right)^2 + \left\lVert\overline{\mathbf{h}}^T_{\text{B}}\mathbf{W}\right\rVert^2, \label{pB1_k}\\
&p_{\text{B}, 2} = \left\lVert\overline{\mathbf{h}}^T_{\text{B}}\mathbf{W}\right\rVert^2, \label{pB2_k} \\
&c^k_{\text{E}, 1} = \log_2\left(\pi e\left(\frac{1}{3}\left(\left(h_{\text{E}, k}v\right)^2 + \left\lVert\overline{\mathbf{h}}^T_{\text{E}, k}\mathbf{W}\right\rVert^2\right) + \overline{\sigma}^2_{\text{E}, k}\right)\right), \label{cE1_k}
\end{align}
\begin{align}
&c^k_{\text{E}, 2} = \log_2 \left(2{\left\lVert\overline{\mathbf{h}}^T_{\text{E}, k}\mathbf{W}\right\rVert^2} + \pi e\overline{\sigma}^2_{\text{E}, k}\right), \label{cE2_k}\\
&p^k_{\text{E}, 1} = \left(h_{\text{E}, k}v\right)^2 + \left\lVert\overline{\mathbf{h}}^T_{\text{E}, k}\mathbf{W}\right\rVert^2, \label{pE1_k}\\
&p^k_{\text{E}, 2} = \left\lVert\overline{\mathbf{h}}^T_{\text{E}, k}\mathbf{W}\right\rVert^2 \label{pE2_k}. 
\end{align}
and then transform $\mathcal{P}{\bf{10}}(t)$ into the following
\begin{subequations}
\begin{alignat}{2}
{\bf{\mathcal{P}11}}(t) \hspace{5mm} & \text{find}   \hspace{2mm}  v, \mathbf{W}, c_{\text{B}, 1}, c_{\text{B}, 2}, p_{\text{B}, 1}, p_{\text{B}, 2}, c^k_{\text{E}, 1}, c^k_{\text{E}, 2}, p^k_{\text{E}, 1}, p^k_{\text{E}, 2} \label{P11:obj_func} \\
& \text{subject to} &		  & \nonumber \\
%& & & \widetilde{\Phi}_s(v, \mathbf{W}) \geq t & \forall k = 1, 2, ..., K,  \label{P9:constraint1} \\
&  \frac{1}{2}\left(c_{\text{B}, 1} - c_{\text{B}, 2} - c^k_{\text{E}, 1} + c^k_{\text{E}, 2}\right) \geq t\left(P_{\text{circuit}} + \sum_{N_T} U_{\text{LEDs}}I_{\text{DC}} + \zeta\left(v^2 + \text{tr}\left(\mathbf{W}\mathbf{W}^T\right)\right)\!\!\right),  \label{P11:constraint1} \\
%& \frac{1}{2}\left(c_{\text{B}, 1} - c_{\text{B}, 2} - c^k_{\text{E}, 1} + c^k_{\text{E}, 2}\right)  \geq \xi_{\text{th}},   \label{P11:constraint2} \\
& c_{\text{B}, 1} \leq \log_2\left(2p_{\text{B}, 1} + \pi e \overline{\sigma}^2_{\text{B}}\right), \label{P11:constraint3} \\
& p_{\text{B}, 1} \leq \left(h_{\text{B}}v\right)^2 + \left\lVert\overline{\mathbf{h}}^T_{\text{B}}\mathbf{W}\right\rVert^2, \label{P11:constraint4} \\
& c_{\text{B}, 2} \geq \log_2\left(\pi e \left(\frac{1}{3}p_{\text{B}, 2} + \overline{\sigma}^2_{\text{B}}\right)\right), \label{P11:constraint5} \\
& p_{\text{B}, 2} \geq \left\lVert\overline{\mathbf{h}}^T_{\text{B}}\mathbf{W}\right\rVert^2, \label{P11:constraint6} \\ 
& c^k_{\text{E}, 1} \geq \log_2\left(\pi e\left(\frac{1}{3}p^k_{\text{E}, 1} + \overline{\sigma}^2_{\text{E}, k}\right)\right), \label{P11:constraint7} \\
& p^k_{\text{E}, 1} \geq \left(h_{\text{E}, k}v\right)^2 + \left\lVert\overline{\mathbf{h}}^T_{\text{E}, k}\mathbf{W}\right\rVert^2, \label{P11:constraint8} \\
& c^k_{\text{E}, 2} \leq \log_2\left(2p^k_{\text{E}, 2} + \pi e\overline{\sigma}^2_{\text{E}, k}\right), \label{P11:constraint9} \\
& p^k_{\text{E}, 2} \leq \left\lVert\overline{\mathbf{h}}^T_{\text{E}, k}\mathbf{W}\right\rVert^2, \label{P11:constraint10} \\ 
& \eqref{P8:constraint1},~ \eqref{P8:constraint2}. \nonumber 
\end{alignat}
\label{P11-parametric}
\end{subequations}
It is seen that constraint \eqref{P11:constraint3}, \eqref{P11:constraint6}, \eqref{P11:constraint8}, \eqref{P11:constraint9}, \eqref{P8:constraint1}, and \eqref{P8:constraint2} are convex while the rest are not. Using the first-order Taylor expansion to approximately convexify those non-convex constraints results in
\begin{align}
\left(h_{\text{B}}v\right)^2 + \left\lVert\overline{\mathbf{h}}^T_{\text{B}}\mathbf{W}\right\rVert^2 \geq  & \left(h_{\text{B}}v^{(j-1)}\right)^2 + 2h^2_{\text{B}}v^{(j-1)}\left(v^{(j)}  - v^{(j-1)}\right)
 + \left\lVert\overline{\mathbf{h}}^T_{\text{B}}\mathbf{W}^{(j-1)}\right\rVert^2  \nonumber \\ & + 
 2\overline{\mathbf{h}}^T_{\text{B}}\left(\mathbf{W}^{(j)} - \mathbf{W}^{(j-1)}\right)\left[\mathbf{W}^{(j-1)}\right]^T\overline{\mathbf{h}}_{\text{B}},
\label{TaylorApprox11-1}
\end{align}
\begin{align}
\log_2\left(\pi e \left(\frac{1}{3}{p_{\text{B,2}}} + \overline{\sigma}^2_{\text{B}}\right)\right)  \leq \log_2\left(\!\pi e \left(\frac{1}{3}{p^{{(j-1)}}_{\text{B}, 2}} + \overline{\sigma}^2_{\text{B}}\right)\!\right)  + \frac{p^{{(j)}}_{\text{B}, 2} - p^{{(j-1)}}_{\text{B}, 2}}{\ln(2)\left(\!\pi e \left(\frac{1}{3}{p^{{(j-1)}}_{\text{B}, 2}} + \overline{\sigma}^2_{\text{B}}\right)\!\right)},
\label{TaylorApprox11-2}
\end{align}
\begin{align}
\log_2\left(\pi e\left(\frac{1}{3}p^k_{\text{E}, 1} +  \overline{\sigma}^2_{\text{E}, k}\right)\right)  \leq \log_2\left(\!\pi e\left(\frac{1}{3}p^{k^{(j-1)}}_{\text{E}, 1} +  \overline{\sigma}^2_{\text{E}, k}\!\right)\right)  + \frac{p^{k^{(j)}}_{\text{E}, 1} - p^{k^{(j-1)}}_{\text{E}, 1}}{\ln(2)\left(\pi e\left(\frac{1}{3}p^{k^{(j-1)}}_{\text{E}, 1} +  \overline{\sigma}^2_{\text{E}, k}\right)\right)},
\label{TaylorApprox11-3}
\end{align}
\begin{align}
\left\lVert\overline{\mathbf{h}}^T_{\text{E}, k}\mathbf{W}\right\rVert^2 \geq  \left\lVert\overline{\mathbf{h}}^T_{\text{E}, k}\mathbf{W}^{(j-1)}\right\rVert^2  +  2\overline{\mathbf{h}}^T_{\text{E}, k}\left(\mathbf{W}^{(j)} - \mathbf{W}^{(j-1)}\right)\left[\mathbf{W}^{(j-1)}\right]^{T}\overline{\mathbf{h}}_{\text{E}, k},
\label{TaylorApprox11-4}    
\end{align}
where $v^{(j)}$, $\mathbf{W}^{(j)}$, $p^{{(j)}}_{\text{B}, 2}$, and $p^{k^{(j)}}_{\text{E}, 1}$ are the solutions to $v$, $\mathbf{W}$, $p_{\text{B}, 2}$, and $p_{\text{E}, 1}$ at the $j$-th iteration of the CCP procedure. Replacing the above approximating terms to \eqref{P11:constraint4}, \eqref{P11:constraint5}, \eqref{P11:constraint7}, and \eqref{P11:constraint10}, respectively, resulting in  
\begin{subequations}
\begin{alignat}{2}
{\bf{\mathcal{P}12}}(t) \hspace{5mm} & \text{find}   \hspace{2mm}  v^{(j)}, \mathbf{W}^{(j)}, c_{\text{B}, 1}, c_{\text{B}, 2}, p_{\text{B}, 1}, p^{{(j)}}_{\text{B}, 2}, c^k_{\text{E}, 1}, c^k_{\text{E}, 2}, p^{k^{(j)}}_{\text{E}, 1}, p^k_{\text{E}, 2} \label{P12:obj_func} \\
& \text{subject to} &		  & \nonumber \\
%& & & \widetilde{\Phi}_s(v, \mathbf{W}) \geq t & \forall k = 1, 2, ..., K,  \label{P9:constraint1} \\
%&  \frac{1}{2}\left(c^k_{\text{B}, 1} - c^k_{\text{B}, 2} - c^k_{\text{E}, 1} + c^k_{\text{E}, 2}\right) \geq t\left(P_{\text{circuit}} + \sum_{N_T} U_{\text{LEDs}}I_{\text{DC}} + \zeta\left(v^2 + \sum_{k = 1}^K\lVert\mathbf{W}\rVert^2\right)\!\!\right),  \label{P11:constraint1} \\
%& \frac{1}{2}\left(c^k_{\text{B}, 1} - c^k_{\text{B}, 2} - c^k_{\text{E}, 1} + c^k_{\text{E}, 2}\right)  \geq \xi_{\text{th}},   \label{P11:constraint2} \\
%& c^k_{\text{B}, 1} \leq \log_2\left(2p^k_{\text{B}, 1} + \pi e \overline{\sigma}^2_{\text{B}}\right), \label{P11:constraint3} \\
& p_{\text{B}, 1} \leq \left(h_{\text{B}}v^{(j-1)}\right)^2 + 2h^2_{\text{B}}v^{(j-1)}\left(v^{(j)}  - v^{(j-1)}\right)
 + \left\lVert\overline{\mathbf{h}}^T_{\text{B}}\mathbf{W}^{(j-1)}\right\rVert^2  \\ & \nonumber  \hspace{11mm} +  2\overline{\mathbf{h}}^T_{\text{B}}\left(\mathbf{W}^{(j)} - \mathbf{W}^{(j-1)}\right)\left[\mathbf{W}^{(j-1)}\right]^T\overline{\mathbf{h}}_{\text{B}},  \\
& c_{\text{B}, 2} \geq  \log_2\left(\!\pi e \left(\frac{1}{3}{p^{{(j-1)}}_{\text{B}, 2}} + \overline{\sigma}^2_{\text{B}}\right)\!\right)  + \frac{p^{{(j)}}_{\text{B}, 2} - p^{{(j-1)}}_{\text{B}, 2}}{\ln(2)\left(\!\pi e \left(\frac{1}{3}{p^{{(j-1)}}_{\text{B}, 2}} + \overline{\sigma}^2_{\text{B}}\right)\!\right)},\label{P12:constraint2} \\
%& p^k_{\text{B}, 2} \geq \left(\overline{\mathbf{h}}^T_{\text{B}}\mathbf{W}\right)^2, \label{P11:constraint6} \\ 
& c^k_{\text{E}, 1} \geq  \log_2\left(\!\pi e\left(\frac{1}{3}p^{k^{(j-1)}}_{\text{E}, 1} +  \overline{\sigma}^2_{\text{E}, k}\!\right)\right)  + \frac{p^{k^{(j)}}_{\text{E}, 1} - p^{k^{(j-1)}}_{\text{E}, 1}}{\ln(2)\left(\pi e\left(\frac{1}{3}p^{k^{(j-1)}}_{\text{E}, 1} +  \overline{\sigma}^2_{\text{E}, k}\right)\right)},\label{P12:constraint3} \\
%& c^k_{\text{E}, 2} \leq \log_2\left(\pi e \left(\frac{1}{3}p^k_{\text{E}, 2} + \overline{\sigma}^2_{\text{E}}\right)\right), \label{P11:constraint9} \\
& p^k_{\text{E}, 2} \leq   \left\lVert\overline{\mathbf{h}}^T_{\text{E}, k}\mathbf{W}^{(j-1)}\right\rVert^2  +  2\overline{\mathbf{h}}^T_{\text{E}, k}\left(\mathbf{W}^{(j)} - \mathbf{W}^{(j-1)}\right)\left[\mathbf{W}^{(j-1)}\right]^{T}\overline{\mathbf{h}}_{\text{E}, k},\label{P12:constraint4} \\
& \eqref{P11:constraint1}, %\eqref{P11:constraint2},
\eqref{P11:constraint3},~\eqref{P11:constraint6}, ~\eqref{P11:constraint8},~\eqref{P11:constraint9},~\eqref{P8:constraint1},~\eqref{P8:constraint2}, \nonumber 
\end{alignat}
\label{P12-parametric}
\end{subequations}
which is a convex feasibility problem, hence can be efficiently handled. The feasibility of $\mathcal{P}\textbf{11}(t)$ can then be determined if there exists a solution to $\mathcal{P}\textbf{12}(t)$. 
Now, it is observed that the solution to $\mathcal{P}\textbf{9}$ is the maximum value of $t$ at which $\mathcal{P}\textbf{11}(t)$ is feasible. Therefore, the bisection method, as described in the $\textbf{Algorithm 3}$ on top of the next page,  can be utilized to solve $\mathcal{P}\bf{9}$. \begin{algorithm2e}[ht]
\caption{Bisection method for solving $\mathcal{P} \textbf{9}$.}
Choose $t_1$ and $t_2$ ($t_1 < t_2$) so that $\mathcal{P} \textbf{12}(t)$ is feasible at $t = t_1$ and is infeasible at $t = t_2$. \\
Choose the error tolerance $\epsilon_3  > 0$. \\
\While{$t_2 - t_1 > \epsilon_3$}{
$t \leftarrow \frac{t_1 + t_2}{2}$.\\
Solve the feasibility problem $\mathcal{P}\textbf{12}(t)$. \\
\eIf{{\rm{there exists a solution to}} $\mathcal{P}{\bf{12}}(t)$}{
$t_1 \leftarrow t$.
}{
$t_2 \leftarrow t$.}
}
\end{algorithm2e}
\section{Numerical Results and Discussions}
In this section, simulation results are presented to demonstrate the SEE of the AN designs for the two transmission schemes. A typical room of the size of 5 m (Length) $\times$ 5 m (Width) $\times$ 3 m (Height) equipped with 4 LED luminaries as shown in Figs.~\ref{sysmodel1} and \ref{sysmodel2} is considered. Bob's and Eves' receivers are placed 0.5 m above the floor. Locations of LED luminaires, Bob, and Eves are determined through a Cartesian coordinate system whose origin is the center of the floor. We also assume that $I_{\text{min}} = 0$ and $I_{\text{max}} \geq 2I_{\text{DC}}$, resulting in $\Delta_{\text{DC}} = I_{\text{DC}}$. Without otherwise noted, $\delta_{\text{B}} = 0$ dB, $P_{\text{circuitry}} = 8$ Watts, $U_{\text{DC}} = 3.3$ Volts, $\zeta = 2$, and error tolerances $\epsilon_{1} = \epsilon{_2} = \epsilon_3 = 10^{-3}$ are used \cite{Duong2021}. Other system parameters are given in Table \ref{table2}. Furthermore, simulation results are averaged over 2000 randomly generated channel realizations for Bob and Eves. Without loss of generality, for the fixed AN-aided SISO scheme, luminaire 1 is chosen to be Alice, while luminaires 2, 3, and 4 act as jammers. Since the information-bearing and AN symbols are both assumed to be zero-mean, the average emitted optical power of each luminaire is dependent solely on $I_{\text{DC}}$ and given as $\overline{p}_t = \eta I_{\text{DC}}$.  
% \begin{figure}[htbp]
% \centering
% \begin{subfigure}[b]{0.48\textwidth}
%     \includegraphics[scale = 0.30]{Figs/Fig2_aNew.eps}
%     \caption{EE performance.}
%     \label{EE}
% \end{subfigure} 
% \begin{subfigure}[b]{0.48\textwidth}
%     \includegraphics[scale = 0.30]{Figs/Fig2_bNew.eps}
%     \caption{SINR performance.}
%     \label{SINR}
% \end{subfigure} 
% \caption{Comparisons between the proposed scheme and the no AN design.}
% \label{Comparison}
% \end{figure}
% \begin{figure}
% \centering
% \includegraphics[scale = 0.5]{Figs/Fig3_New.eps}
% \caption{An illustration of the proposed jamming scheme.}
% \label{Convergence}
% \end{figure}
\subsection{Unknown Eves' CSI}
Firstly, it is necessary to specify the minimum allocated AN power $P_\text{th}$ as described in \eqref{P2:constraint2}. According to the constraint in \eqref{P1:constraint2}, the maximum allowable AN power in the selective AN-aided SISO scheme is $(N_T - 1)I^2_{\text{DC}}$ while it is $N_TI^2_{\text{DC}}$ in the case of the AN-aided MISO transmission. Thus, to avoid the infeasibility of \eqref{P1:constraint2} in both AN transmission schemes, $P_{\text{th}}$ can be chosen to be  $P_{\text{th}} = (N_T - 1)(\rho I_{\text{DC}})^2$, where $\rho \in [0,~1]$ is a constant that controls the maximum amount of the AN power given a fixed value of $I_{\text{DC}}$.
\begin{table}[ht]
 \caption{System Parameters} %title of the table
 \centering % centering table
 \begin{tabular}{l l l} % creating eight columns
 \midrule\midrule
 \multicolumn{2}{c}{\bf{Room and LED configurations}} \\
 \midrule\midrule 
 %Room dimension \\ (Length $\times$ Width $\times$ Height) &  5 (m) $\times$ 5 (m) $\times$ 3 (m) \\  
 % \midrule     
 LED luminaire  positions & luminaire  1 : ($-\sqrt{2}$, $-\sqrt{2}$, 3) ~~ &luminaire  2 : ($\sqrt{2}$, $-\sqrt{2}$, 3)\\ & luminaire  3 : ($\sqrt{2}$, $\sqrt{2}$, 3) ~~ &luminaire  4 : ($-\sqrt{2}$, $\sqrt{2}$, 3) \\
 \midrule     
 LED bandwidth, $B_{\text{mod}}$ & 20 MHz \\ 
 \midrule     
 LED beam angle, $\phi$ & $120^\circ$ \\ (LED Lambertian order is 1) \\
 \midrule 
 LED conversion factor, $\eta$ & 0.44 W/A  \\
 \midrule \midrule    
 \multicolumn{2}{c}{\bf{Receiver photodetectors}} \\
 \midrule \midrule    
 Active area, $A_r$ & 1 $\text{cm}^2$ \\ 
 \midrule     
 Responsivity, $\gamma$ & 0.54 A/W\\ 
 \midrule     
 Field of view (FoV), $\Psi$ & $60^\circ$\\ 
 \midrule     
 Optical filter gain, $T_s(\psi)$ & 1\\ 
 \midrule     
 Refractive index of the concentrator, $\kappa$ & 1.5\\ 
 \midrule \midrule
 \multicolumn{2}{c}{\bf{Other parameters}} \\
 \midrule \midrule
 Ambient light photocurrent, $\chi_{\text{amp}}$ & 10.93 $\text{A}/(\text{m}^2 \cdot \text{Sr}$) \\
 \midrule 
 Preamplifier noise current density, $i_{\text{amp}}$ & 5 $\text{pA}/\text{Hz}^{-1/2}$ \\
 \midrule \midrule
 \end{tabular}
 \label{table2}
 \end{table} 
\begin{figure*}[ht]
    \centering
    \begin{subfigure}[b]{0.32\textwidth}
        \centering
        \includegraphics[width = \textwidth, height = 5.0cm]{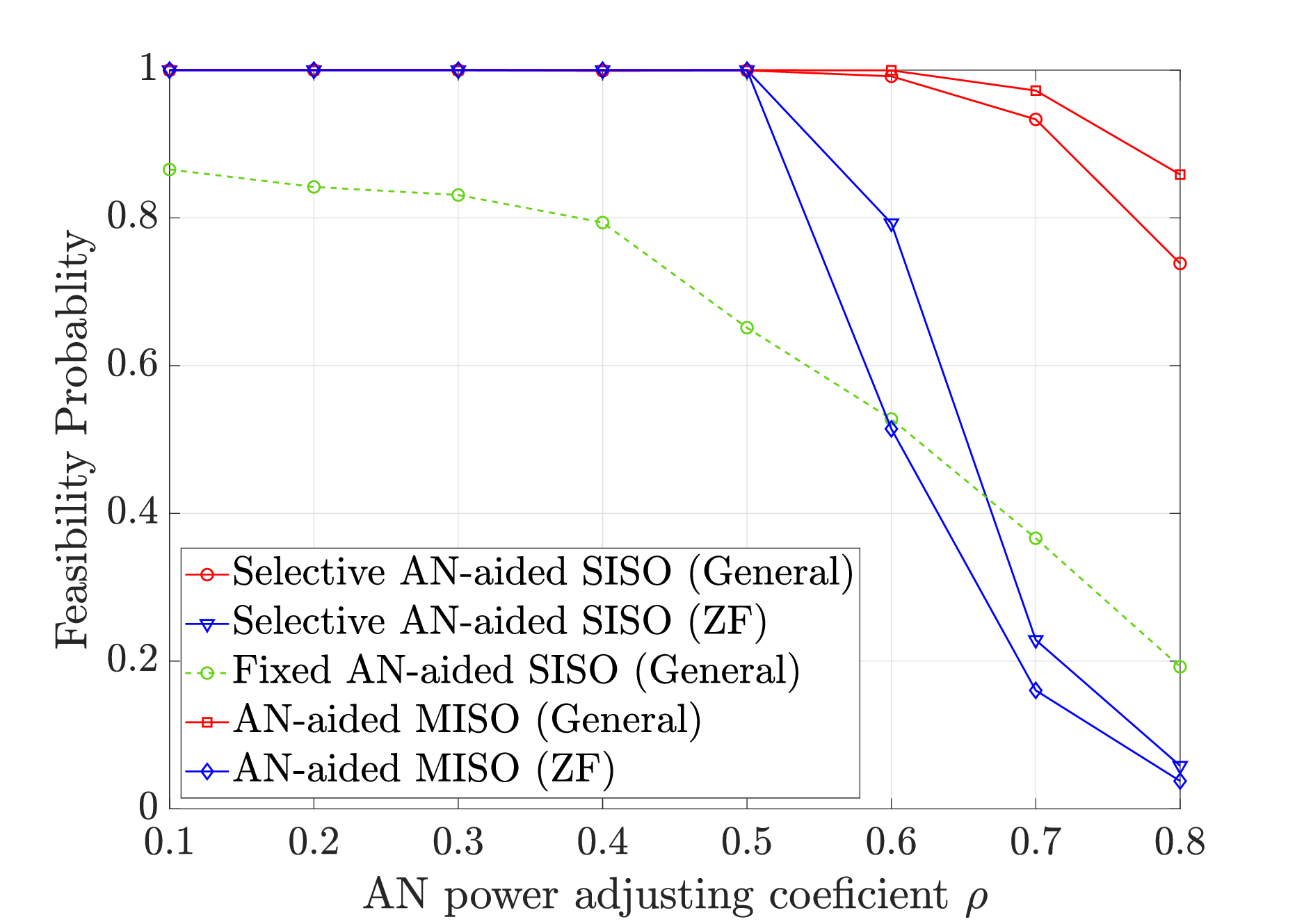}
        \caption{$\overline{p}_t$ = 30 dBm.}
        \label{Feasibility-30}
    \end{subfigure}
    \begin{subfigure}[b]{0.32\textwidth}
        \centering
        \includegraphics[width = \textwidth, height = 5.0cm]{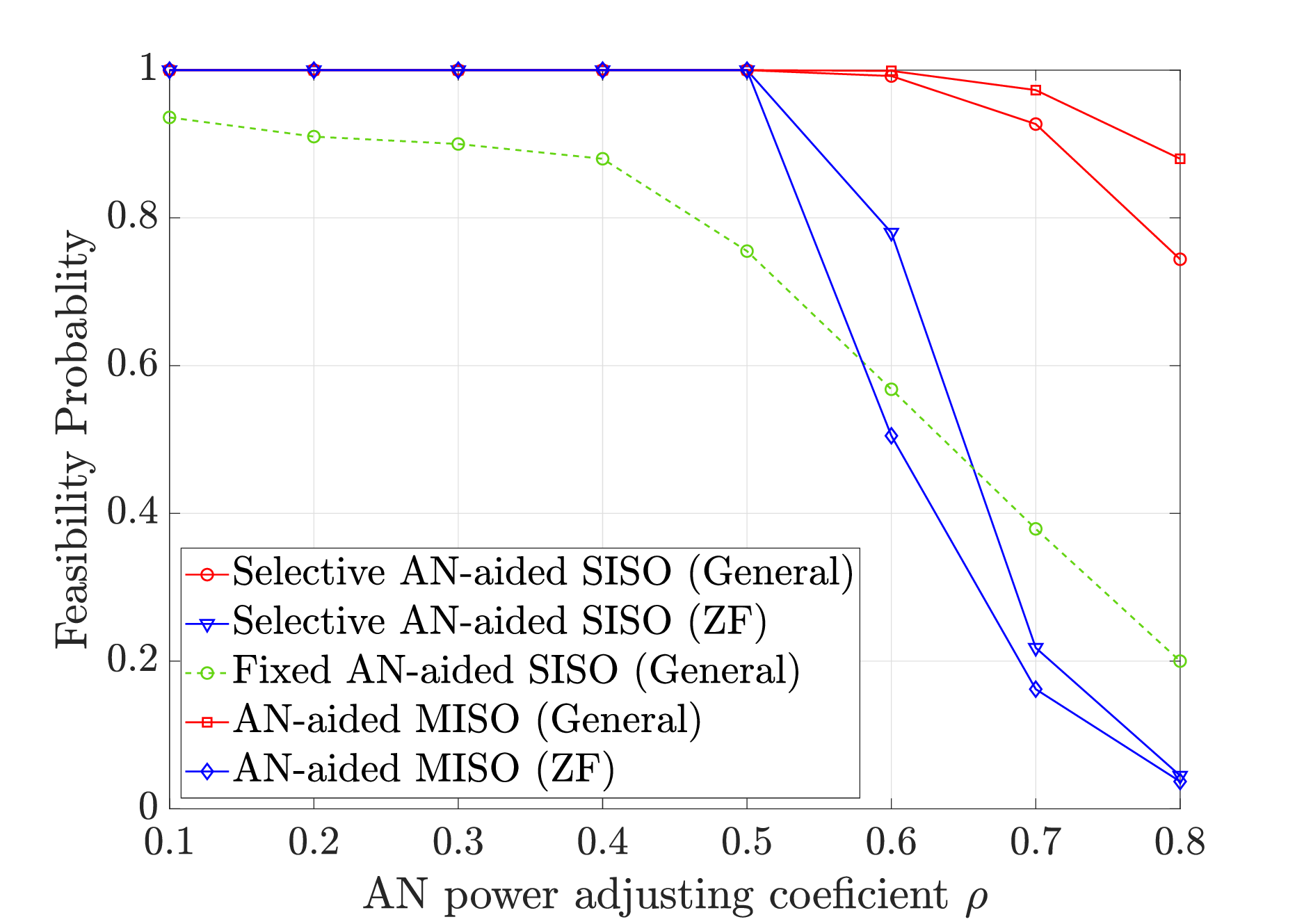}
        \caption{$\overline{p}_t$ = 35 dBm.}
        \label{Feasibility-35}
    \end{subfigure}
    \begin{subfigure}[b]{0.32\textwidth}
        \centering
        \includegraphics[width = \textwidth, height = 5.0cm]{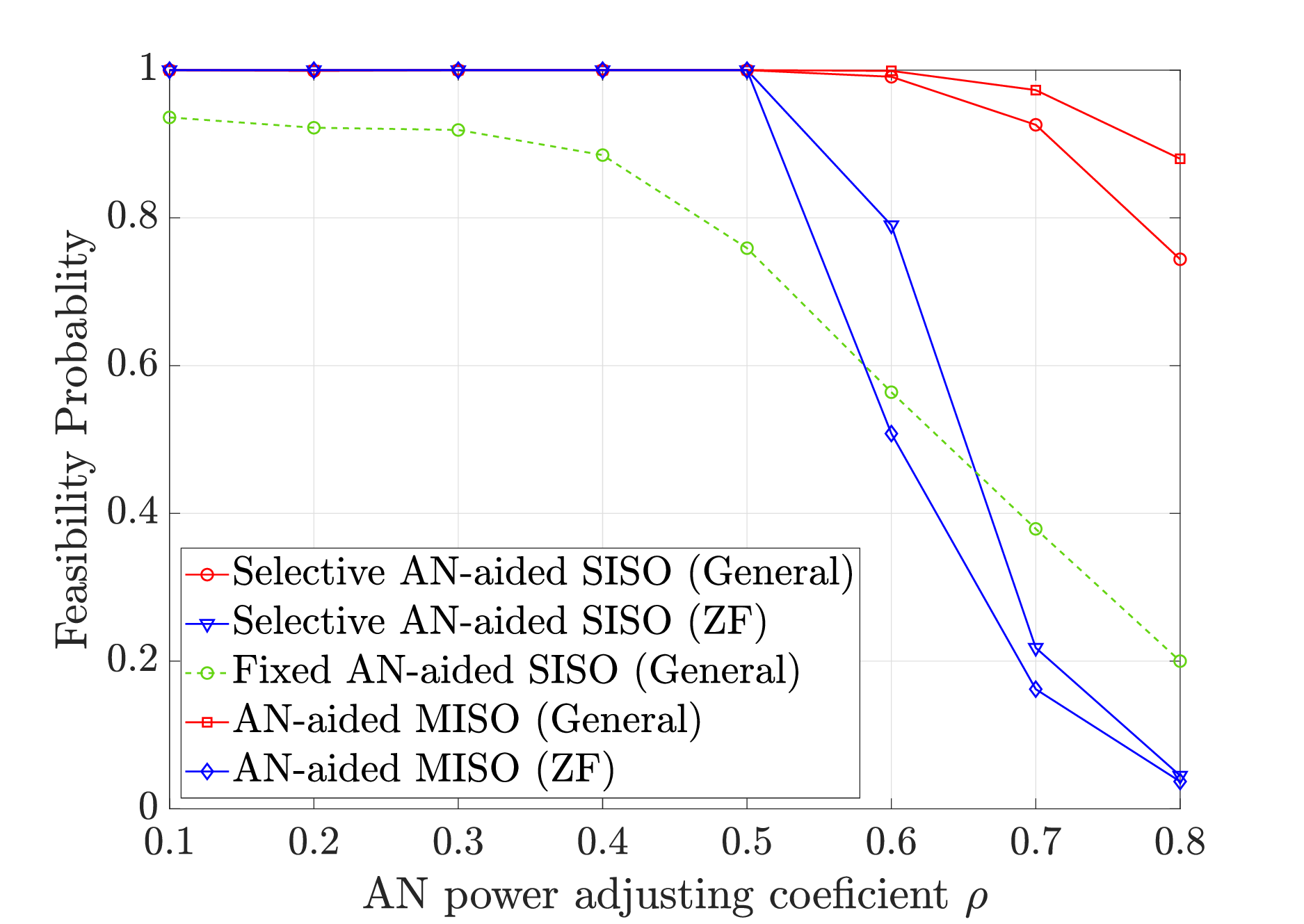}
        \caption{$\overline{p}_t$ = 40 dBm.}
        \label{Feasibility-40}
    \end{subfigure}
    \caption{Feasibility probability of AN transmission schemes.}
    \label{Fesibility_vs_ANPower}
\end{figure*}

Recall that depending on the choice of $\delta_{\text{B}}$ and $\rho$, constraints \eqref{P2:constraint1} and \eqref{P2:constraint2} may not be simultaneously satisfied, which results in an infeasibility of $\mathcal{P}\textbf{2}$ (similar arguments for the fixed AN-aided SISO and AN-aided MISO schemes). Since $\delta_{\text{B}} = 0$ dB is assumed, we show in Figs.~\ref{Feasibility-30}, \ref{Feasibility-35}, and \ref{Feasibility-40} the impact of $\rho$ on the feasibility of $\mathcal{P}\textbf{2}$ for different values of $\overline{p}_t$. When $\rho$ is less than 0.5, except for the fixed SISO scheme, all AN transmission schemes are always feasible. Their feasibility probabilities, however, start decreasing as $\rho$ exceeds 0.5, especially for the ZF AN designs (due to their reduced feasible regions imposed by the ZF constraint). Also, the low feasibility probability of the fixed SISO scheme is due to the fact that Bob can be out of the service zone as a fixed luminaire is chosen 
to transmit the information-bearing signal. To ensure sufficiently high feasibility probabilities of all transmission schemes, in the following, $\rho = 0.2$ is chosen for simulations.

To qualitatively compare the SEE performance of the examined AN-aided transmission schemes, it suffices to consider the case that the number of Eves is 1. Figs.~\ref{EE_vs_Power} and \ref{SEE_vs_Power} depict the EE of Bob's channel (i.e., the solution to $\mathcal{P}\textbf{1}$) and the resultant SEE with respect to the average luminaire optical power $\overline{p}_t$. It is clearly shown that when only Bob's channel is considered, the MISO transmission performs better than the selective and fixed SISO. This verifies the benefit of spatial precoding in improving EE. Interestingly, the resultant SEE shows a contrary observation where the proposed selective SISO scheme significantly outperforms the MISO and the fixed SISO transmissions. For instance, the optimal SEE of the selective SISO is more than twice that of the MISO scheme and about 80\% higher than that of the fixed SISO scheme. 
\begin{figure*}[ht]
    \begin{subfigure}{0.5\textwidth}
        \includegraphics[width = \textwidth, height = 0.70\textwidth]{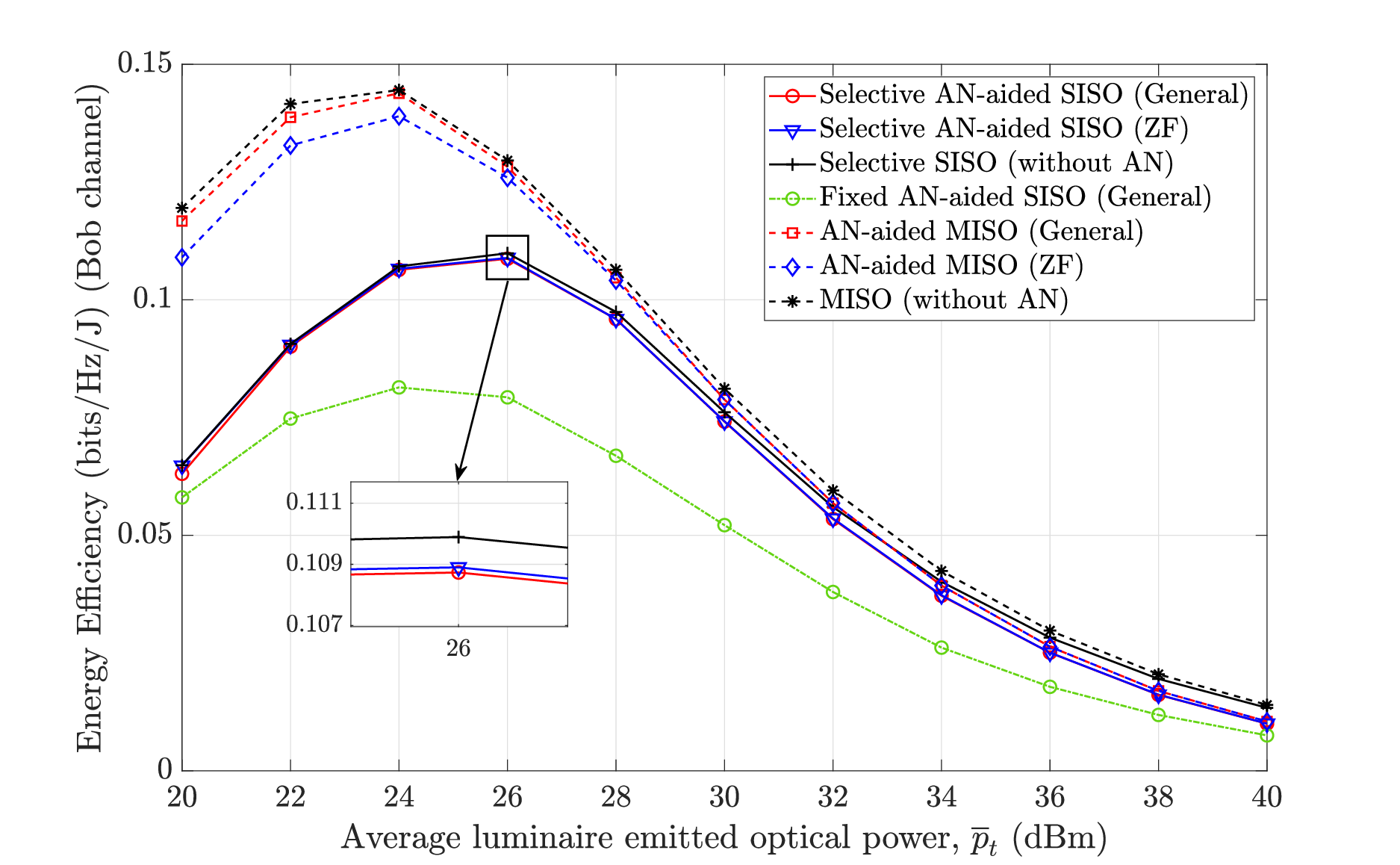}
        \caption{EE of Bob's channel versus $\overline{p}_t$.}
        \label{EE_vs_Power}
    \end{subfigure}
    \begin{subfigure}{0.5\textwidth}
        \includegraphics[width = \textwidth, height = 0.70\textwidth]{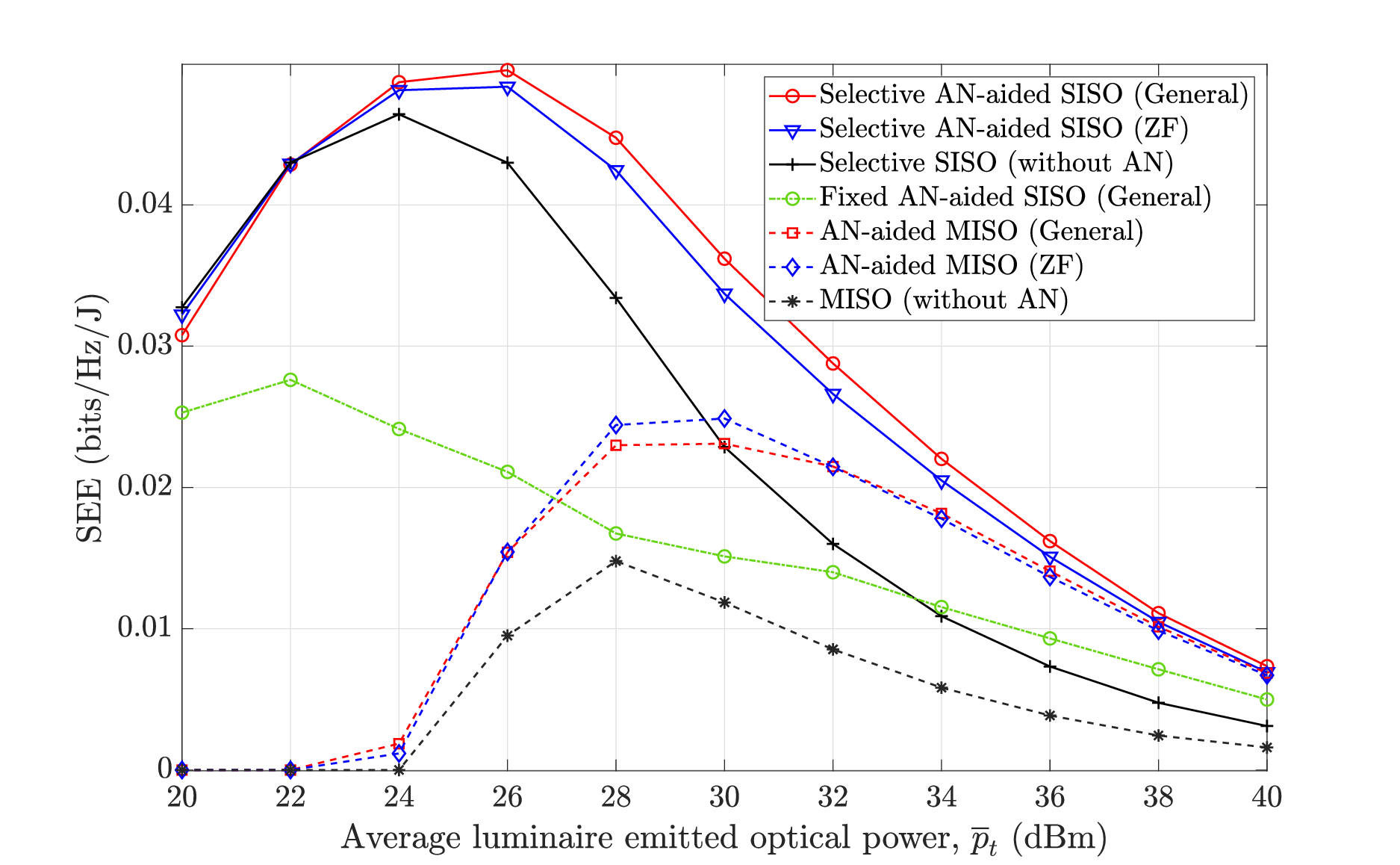}
        \caption{SEE versus $\overline{p}_t$.}
        \label{SEE_vs_Power}
    \end{subfigure}\hfill
    \begin{subfigure}{1\textwidth}
        \includegraphics[width = \textwidth, height = 0.4\textwidth]{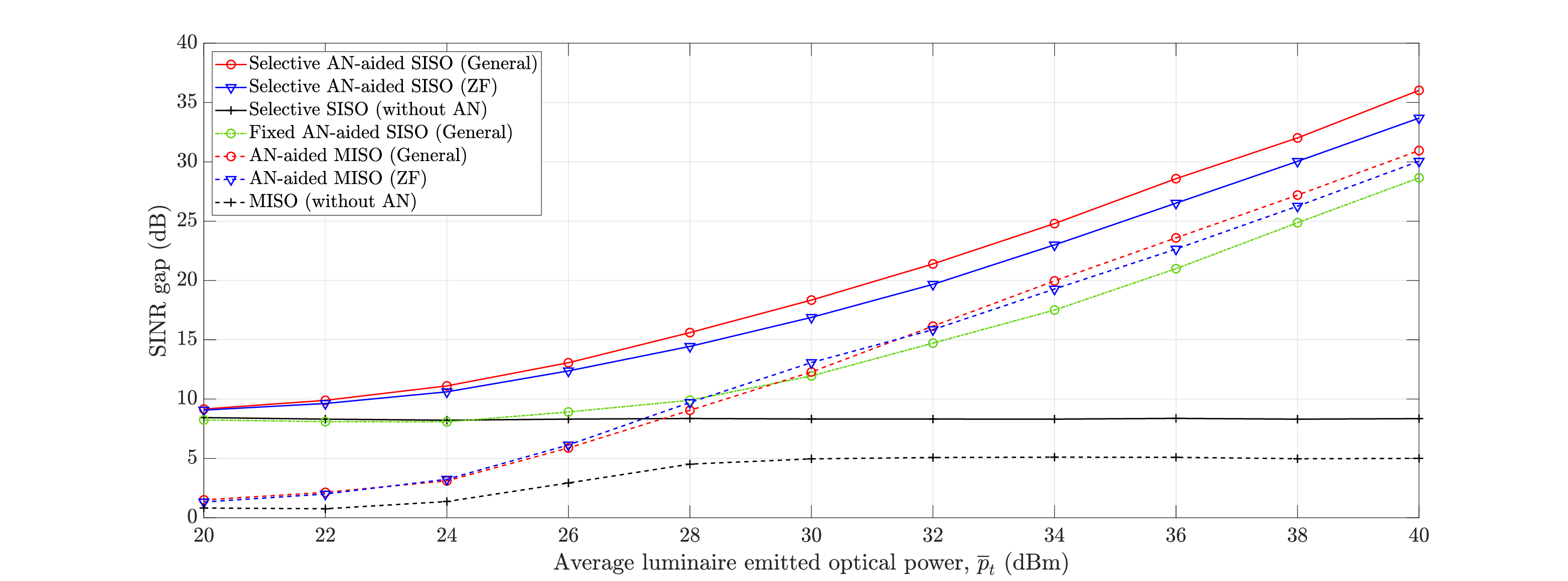}
        \caption{SINR gap versus $\overline{p}_t$.}
        \label{SINR_vs_Power}
    \end{subfigure}
    \caption{Performance of AN transmission schemes in the case of unknown Eves' CSI.}
\end{figure*}
This can be explained in Fig.~\ref{SINR_vs_Power}, which plots the gap between the signal-to-interference-plus-noise ratios (SINR) of Bob's and Eve's channels. It is seen that the SINR gap in the case of the MISO scheme is smaller than that of the selective SISO, especially in the low optical power region. Recall that the AN design, in this case, tries to maximize the EE of Bob's channel without consideration of Eve due to its unknown CSI (i.e., indirect AN design). Consequently, using more luminaires for transmitting the information-bearing signal enables Eve to receive more signal power, which results in a lower secrecy rate compared to that of the selective SISO scheme. %As the maximum allowable AN power becomes significant by increasing $\overline{p}_t$ $($because $P_{\text{th}} = (N_t - 1)(\rho I_{\text{DC}})^2$$)$, the SINRs of Eve's channel in AN-aided transmission schemes start to decrease, leading to an increase in the SINR gap. 
As the transmit optical power $\overline{p}_t$ increases, an increase in the SINR gap is observed. Nevertheless, this improvement (due to increased $\overline{p}_t$) does not result in a better SEE because of the dominance of the power consumption. 
It is also evident that AN-aided transmission schemes offer better SEE performance than their no-AN (i.e., without AN) counterparts. Regarding the selective SISO transmissions, the ZF AN-aided design, despite its low complexity, archives a comparable SEE with that of the general design. Hence, it may be preferable for the system design when reducing the signal processing time is of importance.

We now discuss the complexity of the general AN designs for the selective SISO and MISO transmissions. As presented in Sec.~III.A, the solving procedure involves two
iterative algorithms (i.e., Dinkelback-type and CCP-type algorithms). Thus, Figs.~ \ref{Convergence-SISO} and \ref{Convergence-MISO} show the convergence error of the EE with respect to the number of iterations of the Dikelbach-type (\textbf{Algorithm 1}) and CCP-type (\textbf{Algorithm 2}) algorithms for the selective AN-aided SISO and the AN-aided MISO schemes, respectively. 
\begin{figure*}[ht]
    \centering
    \begin{subfigure}[b]{0.49\textwidth}
        \centering
        \includegraphics[width = \textwidth, height = 6.5cm]{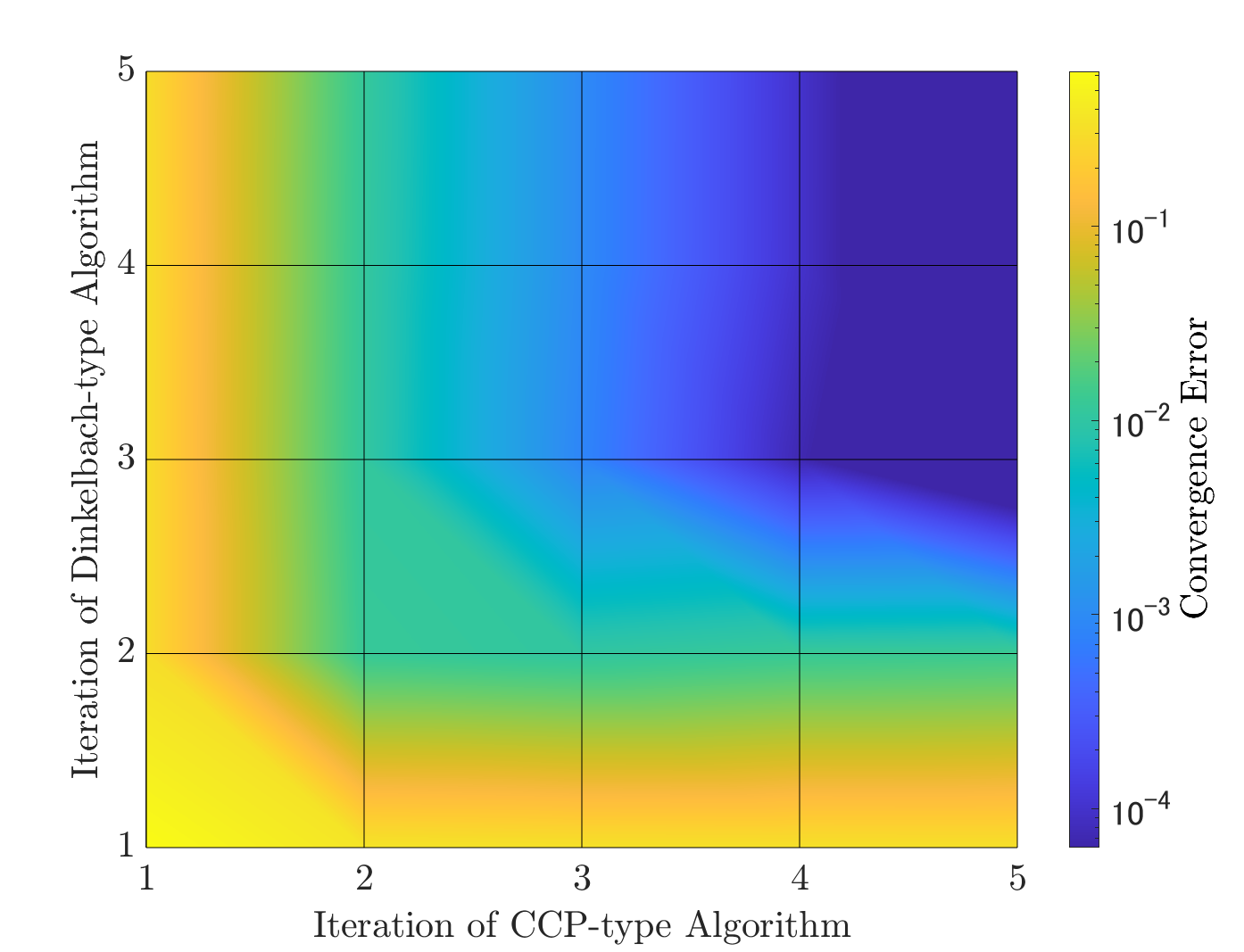}
        \caption{Selective AN-aided SISO.}
        \label{Convergence-SISO}
    \end{subfigure}
    \begin{subfigure}[b]{0.49\textwidth}
        \centering
        \includegraphics[width = \textwidth, height = 6.5cm]{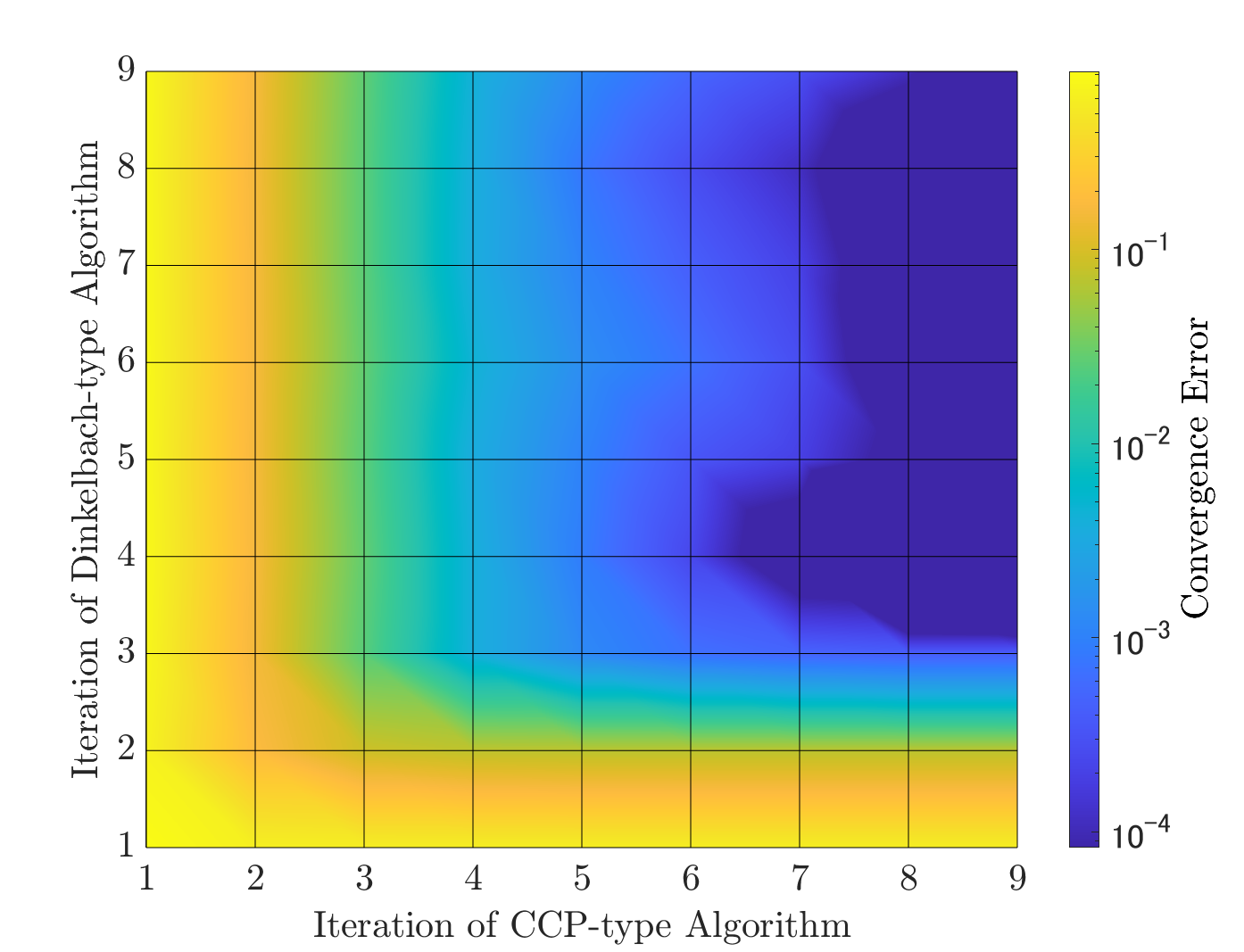}
        \caption{AN-aided MISO.}
        \label{Convergence-MISO}
    \end{subfigure}
    \caption{Convergence behaviors of the two AN schemes.}
    \label{UnknownHe:Secrecy_vs_Power}
\end{figure*}
For the former, there requires, on average, three iterations for both algorithms to achieve the convergence error of $10^{-3}$. For the latter to attain the same convergence error, at least six and four iterations are needed for the CCP-type and the Dinkelbach-type algorithms, respectively. The higher complexity of the AN-aided MISO is due to the larger sizes of its optimizing variables (i.e., $N_T$ vs. 1 for the precoder of the information symbol and $N_T$ vs. $N_T-1$ for the precoder of the AN).
In practical scenarios, when more LED luminaries are deployed (i.e., $N_T$ increases), the complexity may increase substantially due to the bigger size of the optimization variables.
\subsection{Known Eves' CSI}
\begin{figure}[H]
\centering
\begin{minipage}{.49\textwidth}
\includegraphics[width=\linewidth, height = .80\textwidth]{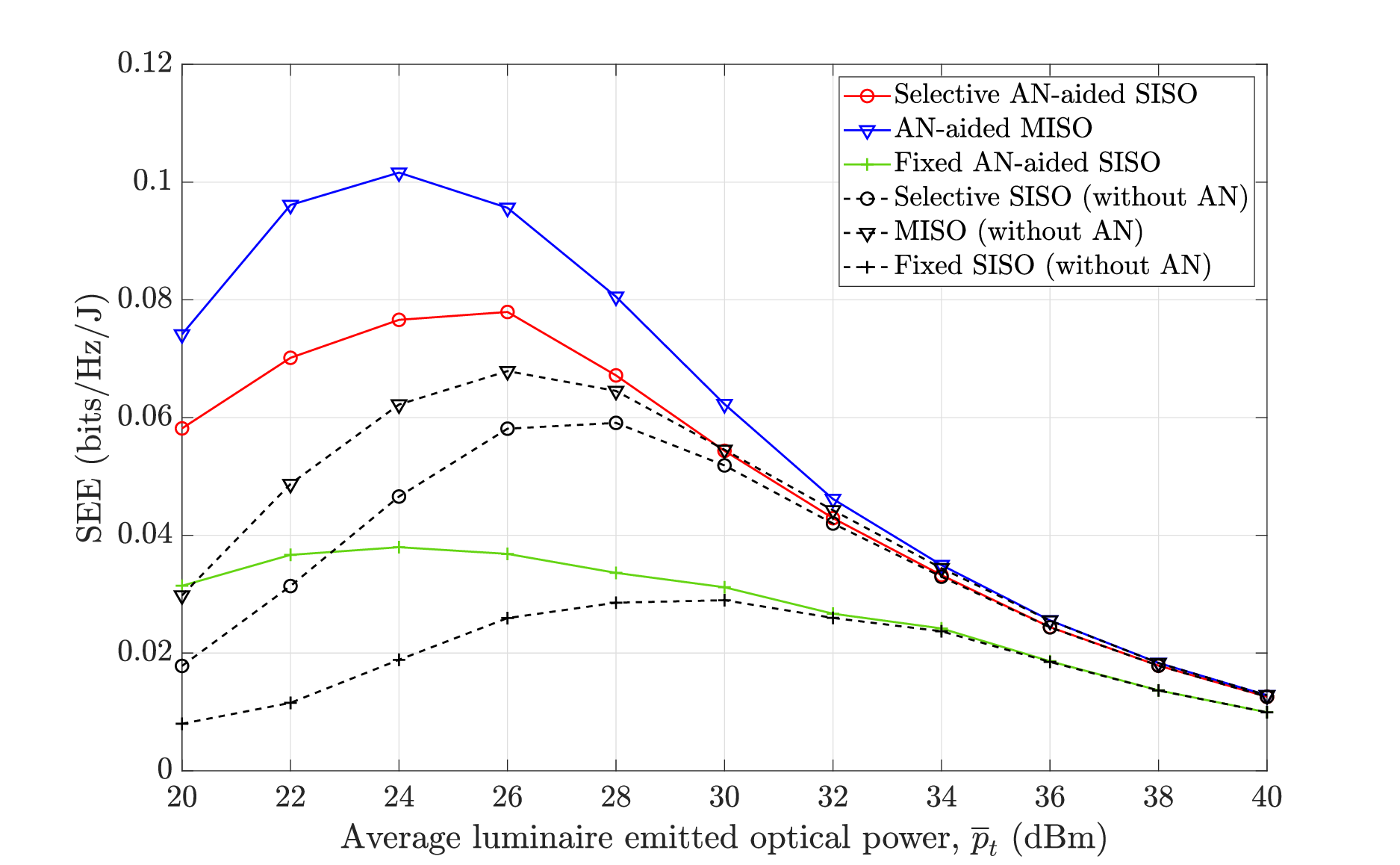}
\caption{SEE versus $\overline{p}_t$.}
\label{SEE_vs_Power1}
\end{minipage}
\begin{minipage}{.49\textwidth}
\includegraphics[width=\linewidth, height = .80\textwidth]{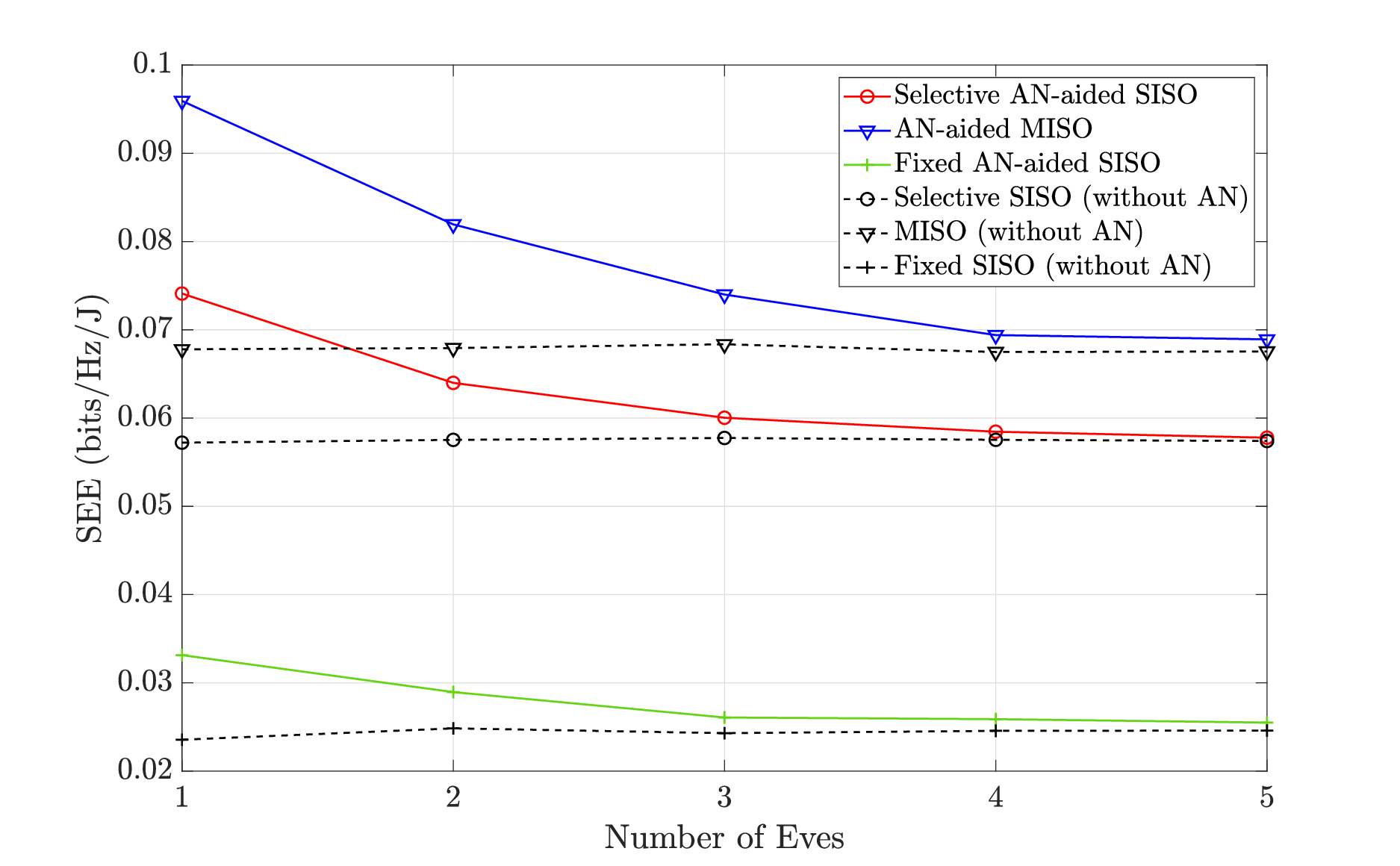}
\caption{SEE versus the number of Eves.}
\label{SEE_vs_NoOfEve}
\end{minipage}
\end{figure}
In Fig.~\ref{SEE_vs_Power1}, comparisons between AN-aided and no-AN transmission schemes in terms of the SEE are illustrated when Eves' CSI is known at the transmitters when the number of Eves is set to 1. In sharp contrast to the case of unknown Eves' CSI, it is shown in this case that at their respective optimal SEE, the MISO scheme outperforms the selective and fixed SISO counterparts by 30\% and 167\%, respectively. In the case of the indirect design when Eve's CSI is not available, due to the absence of Eve's CSI in the AN design, only one luminaire (that provides the highest channel gain) should be used to transmit the information-bearing signal so that as little as possible the information is eavesdropped on by Eves. This strategy is, however, not optimal for the direct design when Eves' CSI is considered. Since AN can be specifically constructed to attenuate Eves' channels, more luminaires should be employed to transmit the information-bearing signal so that the SEE can be improved. 

As we study VLC systems where multiple non-colluding Eves can coexist, Fig.~\ref{SEE_vs_NoOfEve} illustrates the SEE with respect to the number of Eves when $\overline{p}_t =$ 26 dBm. 
It is highlighted that while the SEE in the case of AN-aided transmission considerably deteriorates as the number of Eves increases, it stays unchanged when AN is not used. Specifically, when the number of Eves increases from 1 to 5, the SEE drops 22\%, 28\%, and 23\% for the selective AN-aided SISO, AN-aided MISO, and fixed SISO transmission schemes, respectively. This SEE degradation is due to the proportionality between the number of Eves and the number of AN streams. Hence, when more Eves are present, more power should be allocated to the AN, leading to a reduction in the SEE. Furthermore, it is seen that as the number of Eves increases, the SEE of the AN-aided transmission asymptotes that of the no-AN scheme.    
%In the case of multiple Eves, the SEE is calculated corresponding to the eavesdropper who incurs the lowest achievable secrecy rate. As a result, when the number of Eves increases, the lowest achievable secrecy rate decreases (the probability of an eavesdropper acquiring a high-quality channel increases), leading to a degradation of the SEE in all transmission schemes. 
\section{Conclusion}
In this paper, we have examined the selective AN-aided SISO and AN-aided MSIO transmission schemes and studied their SEE performances. Simulation results revealed that while the selective AN-aided SISO achieves better SEE than the AN-aided MISO scheme when Eves' CSI is unavailable at the transmitting side, it is the opposite when Eves' CSI is available. It was also observed that using AN is advantageous over the no-AN transmission regardless of the availability of Eves' CSI. This work considered the case of a single legitimate user, which is a relatively special assumption. It should be noted that the proposed selective AN-aided SISO scheme was not designed for a more general scenario with multiple Bobs. 
Therefore, our future research attempts a modified transmission scheme to support such a setting. 

%When Eves' CSI is not available at the transmitting side, an indirect AN design, which maximizes the EE of Bob's channel, was considered. In A direct design to maximize the SEE was, no
%In the proposed scheme, the closest LED luminaire to Bob serves as the transmitter while the rest of the luminaries act as jammers. Numerical results showed that while using AN significantly reduces Eves' SINRs, thus improving the physical secrecy, it caused negligible EE loss compared to the case of no AN. It was also demonstrated that the proposed algorithms are efficient in solving the optimization design, as a few iterations were needed to achieve the desired solution.
\begin{appendices}
\section{Proof of Proposition 1}
We first prove that $C_{\text{B}}(v, \mathbf{w})$ is strictly monotonically decreasing with $\left(\overline{\mathbf{h}}^T_{\text{B}}\mathbf{w}\right)^2$. Indeed, the first derivative of $C_{\text{B}}(v, \mathbf{w})$ with respect to $\left(\overline{\mathbf{h}}^T_{\text{B}}\mathbf{w}\right)^2$ is given by
\begin{align}
    C'_{\text{B}}(v, \mathbf{w}) = \frac{\left(2-\frac{\pi e}{3}\right)\pi e\overline{\sigma}^2_{\text{B}} - \frac{2\pi e\left(h_{\text{B}}v\right)^2}{3}}{2^{2C_{\text{B}}(v, \mathbf{w})+1}\log_2e\left(\pi e\right)^2\left(\frac{\left(\overline{\mathbf{h}}^T_{\text{B}}\mathbf{w}\right)^2}{3} + \overline{\sigma}^2_{\text{B}}\right)^2},
\end{align}
which is obviously negative. Thus, the optimal solution $\mathbf{w}^*$ to maximize  $C_{\text{B}}(v, \mathbf{w})$ satisfies  $\overline{\mathbf{h}}_{\text{B}}\mathbf{w}^* = 0$, which corresponds to  $\mathbf{w}^*$ laying on the null-space of $\overline{\mathbf{h}}^T_{\text{B}}$. It is then easily verified that $\mathbf{w}^* = \mathbf{0}$ simultaneously maximizes $C_{\text{B}}(v, \mathbf{w})$ and minimizes $P_{\text{circuit}} + \sum_{N_T} U_{\text{LEDs}}I_{\text{DC}} + \zeta\left(v^2 + \lVert\mathbf{w}\rVert^2\right)$ while satisfying the constraint in \eqref{P1:constraint2}. This proves that $\mathbf{w}^* = \mathbf{0}$ is optimal to $\mathcal{P}\mathbf{1}$.
\section{Proof of Proposition 2}
For expressional simplicity, we denote $a = \frac{2h^2_{\text{B}}}{\pi e \overline{\sigma}^2_{\text{B}}}$ and $b(\phi*) = 2\left(P_{\text{circuit}} + \sum_{N_T}U_{\text{LEDs}}I_{\text{DC}} + \zeta\phi^*\right)$. The derivative of $\Phi_{\text{B}}(V, \phi^*)$ with respective to $V$ is then given by 
\begin{align}
    \Phi'_{\text{B}}(V, \phi*) = \frac{\frac{a(b(\phi^*) + 2\zeta V)}{\ln(2)(1 + aV)} - 2\zeta\log_2(1 + aV)}{\left(b(\phi^*) + 2\zeta V\right)^2}.
\end{align}
Let $f(V, \phi^*) = \frac{a(b(\phi^*) + 2\zeta V)}{\ln(2)(1 + aV)} - 2\zeta\log_2(1 + aV)$. We then have the derivative of $f(V, \phi*)$ with respective to $V$  is $f'(V, \phi^*) = \frac{-a^2}{\ln(2)}\frac{b(\phi^*)+2V}{(1+aV)^2}$. Since $f'(V, \phi^*) < 0$, $f(V, \phi^*)$ is strictly decreasing with $V$. As it is deduced from \eqref{P7:constraint1} and \eqref{P7:constraint2} that $\frac{\delta_{\text{B}}\overline{\sigma}^2_{\text{B}}}{h^2_{\text{B}}} \leq V \leq \Delta^2_{\text{DC}}$, let us consider the following three scenarios.
\begin{itemize}
    \item If $f\left(\frac{\delta_{\text{B}}\overline{\sigma}^2_{\text{B}}}{h^2_{\text{B}}}, \phi^*\right) \geq 0$ and $f(\Delta_{\text{DC}}^2, \phi^*) \leq 0$, a critical point (also the optimal solution) $V^*$ to $\Phi_{\text{B}}(V, \phi^*)$ exists and is unique. Furthermore, $V^*$ can be found using the bisection method \cite{Young2015}. 
    \item If $f(\Delta_{\text{DC}}^2, \phi^*) > 0$ then $\Phi_{\text{B}}(V, \phi^*)$ is strictly increasing with $V$. Therefore, the optimal solution $V^* = \Delta_{\text{DC}}^2$.
    \item If $f\left(\frac{\delta_{\text{B}}\overline{\sigma}^2_{\text{B}}}{h^2_{\text{B}}}, \phi^*\right) < 0$ then $\Phi_{\text{B}}(V, \phi^*)$ is strictly decreasing with $V$. Hence, the optimal solution $V^*$ = $\frac{\delta_{\text{B}}\overline{\sigma}^2_{\text{B}}}{h^2_{\text{B}}}$. 
\end{itemize}
The proof is completed. 
\end{appendices}
\bibliographystyle{ieeetr}
\bibliography{references}
\end{document}